\PassOptionsToPackage{utf8}{inputenc}
\documentclass[conference]{IEEEtran}
\usepackage{natbib}
\usepackage[margin=0.5in]{geometry}

\usepackage{multicol}
\usepackage{amsmath,amssymb,amsthm,amsfonts}
\usepackage{algorithmic}
\usepackage{bm}
\usepackage{dsfont}
\usepackage{booktabs}
\usepackage{numprint}
\usepackage[english]{babel}
\newcommand{\np}[1]{\numprint{#1}}

\usepackage[colorlinks=true,linkcolor=blue,citecolor=blue,urlcolor=blue]{hyperref}

\usepackage{tikz}
\usetikzlibrary{arrows,automata,fit,matrix,backgrounds,positioning}
\usetikzlibrary{decorations.pathmorphing} 

\usepackage{array}
\newcolumntype{H}{>{\setbox0=\hbox\bgroup}c<{\egroup}@{}} 

\usepackage{lastpage,fancyhdr,graphicx, longtable}
\usepackage{epstopdf}
\usepackage[version=3]{mhchem}
\usepackage{multirow}

\usepackage{nameref}
\graphicspath{ {./images/} }

\usepackage{listings}
\newtheorem{lemma}{Lemma}

\newtheorem{proposition}{Proposition}

\usepackage{extarrows}

\newcommand{\ms}{\mathcal{MS}}
\newcommand{\Sp}{\mathcal{S}}
\newcommand{\MS}{\ms}
\newcommand{\pa}{\mathcal{H}}
\newcommand{\upp}[1]{{{#1}^{\uparrow}}}

\newcommand{\Y}{H}
\newcommand{\M}{M}
\newcommand{\tildeM}{\tilde{M}}
\newcommand{\barM}{\bar{M}}
\newcommand{\hatM}{\hat{M}}
\newcommand{\mb}{\M}

\newcommand{\abs}[1]{\left\lvert#1\right\rvert} 
\newcommand{\eqpa}[1]{\approx_{#1}} 
\newcommand{\R}{\mathcal{R}}  
\newcommand{\rr}{\mathbf{rr}}        
\newcommand{\Next}{\ensuremath{\mathit{out}}}        
\newcommand{\MC}{{MC}}        
\newcommand{\reduce}[2]{{{#1}_{#2}}}

\newcommand{\A}{\mathcal{L}} 
\newcommand{\act}[1]{\xlongrightarrow{#1}}          
\newcommand{\Mgt}[3]{{#1}_{{#2}\xrightarrow{}{#3}}}
\newcommand{\rate}{\alpha}

\newtheorem{definition}{Definition}

\newcommand{\Reag}{\ensuremath{\bm{\rho}}} 
\newcommand{\Prod}{\ensuremath{\bm{\pi}}} 

\definecolor{violet(ryb)}{rgb}{0.53, 0.0, 0.69}

\newcommand{\crn}{network}
\newcommand{\ctmc}{Markov chain}
\newcommand{\sbe}{SE}

\newcommand{\se}{\sbe}
\newcommand{\calcium}{\ce{Ca^{2+}}}
\newcommand{\cam}{\ce{CaM}}
\newcommand{\camkii}{\ce{CaMKII}}

\newcommand{\sitext}{\emph{Supplementary Material}}

\renewcommand{\ctmc}{CTMC}

\newcommand{\rn}{RN}

\usepackage{subfig}

\begin{document}


\title{Exact maximal reduction of stochastic reaction networks by species lumping}

\author{
\IEEEauthorblockN{%
Luca Cardelli\IEEEauthorrefmark{1},
Isabel Cristina Perez-Verona\IEEEauthorrefmark{2},
Mirco Tribastone\IEEEauthorrefmark{2}
}
\IEEEauthorblockN{%
Max Tschaikowski\IEEEauthorrefmark{3},
Andrea Vandin\IEEEauthorrefmark{4},
Tabea Waizmann\IEEEauthorrefmark{2}
}
\IEEEauthorblockA{\IEEEauthorrefmark{1}Department of Computer Science, University of Oxford, 34127, UK}
\IEEEauthorblockA{\IEEEauthorrefmark{2}IMT School for Advanced Studies, Lucca, 55100, Italy}
\IEEEauthorblockA{\IEEEauthorrefmark{3}Department of Computer Science, University of Aalborg, 34127, Denmark}
\IEEEauthorblockA{\IEEEauthorrefmark{4}Sant'Anna School of Advanced Studies, Pisa, Italy}}

%
%
%

\maketitle

To whom correspondence should be addressed: Mirco Tribastone.

\begin{abstract}
\textbf{Motivation:} Stochastic reaction networks are a widespread model to describe biological systems where the presence of noise is relevant, such as in cell regulatory processes. Unfortunately, in all but simplest models the resulting discrete state-space representation hinders analytical tractability and makes numerical simulations expensive. Reduction methods can lower complexity by computing model projections that preserve dynamics of interest to the user.\\
\textbf{Results:} We present an exact lumping method for stochastic reaction networks with mass-action kinetics. It hinges on an equivalence relation between the species, resulting in a reduced network where the dynamics of each macro-species is stochastically equivalent to the sum of the original species in each equivalence class, for any choice of the initial state of the system. Furthermore, by an appropriate encoding of kinetic parameters as additional species, the method can establish equivalences that do not depend on specific values of the parameters. The method is supported by an efficient algorithm to compute the largest species equivalence, thus the maximal lumping. The effectiveness and scalability of our lumping technique, as well as the physical interpretability of resulting reductions, is demonstrated in several models of signaling pathways and epidemic processes on complex networks.\\
\textbf{Availability:} The algorithms for species equivalence have been implemented in the software tool ERODE, freely available for download from~\href{https://www.erode.eu}{https://www.erode.eu}.
\end{abstract}

\smallskip

\emph{This article has been submitted to the journal Bioinformatics}

\section{Introduction}

Stochastic reaction networks are a foundational model to study biological systems where the presence of noise cannot be neglected, for instance in cell regulatory processes governed by low-abundance biochemical species~\citep{doi:10.1002/bies.950171112}, which may introduce significant variability in gene expression~\citep{Elowitz16082002}. Their analysis---either by solution of the master equation or by stochastic simulation---is fundamentally hindered by a discrete representation of the state space~\citep{kampen01}, which leads to a combinatorial growth in the number of states in the underlying Markov chain as a function of the abundances of the species.

Here we present a method for exact reduction that preserves the stochastic dynamics of mass-action reaction networks, a fundamental kinetic model in computational systems biology~\citep{10.1371/journal.pcbi.1004012}. The method rests on a relation between species, called \emph{species equivalence} (SE), which can be checked through criteria that depend on the set of reactions of the network. SE gives rise to a reduced stochastic reaction network where the population of each macro-species  tracks the sum of the population levels of all species belonging to an equivalence class.

As with all reduction methods, SE implies some loss of information; namely, the individual dynamical behavior of a species that is aggregated into a macro-species cannot be recovered in general. However, our algorithm for computing SE gives freedom to the modeler as to which original variables to preserve in the reduced network. Indeed, building upon a celebrated result in theoretical computer science~\citep{partitionref}, we compute SE as the coarsest partition that satisfies the equivalence criteria and that refines a given initial partition of species. Thus, a species of interest that is isolated in a singleton block is guaranteed to be preserved in the reduced network. Our partition-refinement algorithm is computationally efficient, in the sense that the algorithm runs in polynomial time as function of the number of species and reactions of the original network. Finally, we can prove the existence of a maximal SE, i.e., the equivalence that leads to the coarsest aggregation of the reaction network.

Formally, SE can be seen a lifting to reaction networks of the notion of lumpability of Markov chains~\citep{Ke76,BuchholzOrdinaryExact}. That is, the reduced network yields a state space where each macro-state tracks the sum of the probabilities of the states in the original Markov chain. Ordinary lumpability requires the availability of the state space that underlies the master equation~\citep{kampen01}; thus, it also requires the initial state of the Markov chain  to be fixed. Instead, SE works at the structural level of the reaction network, by lumping species instead of states; thus, it involves the analysis of an exponentially smaller mathematical object in general. In addition, a practically useful consequence of reasoning at the network level is that an SE holds for any initial state. 
Given that a reaction network can be seen as a Petri net where each species is represented as a place~\citep{petriCRN}, our structural approach is close in spirit to the notion of place bisimulation~\citep{Autant:1992aa}. However, that induces a bisimulation over markings in the  classical, non-quantitative sense~\citep{JOYAL1996164}.

There are several methods for the reduction of the deterministic rate equations of biochemical reaction networks, e.g.,~\citet{Snowden2017}. However, these reductions do not preserve the stochastic behavior in general. For stochastic models in systems biology, lumpability has been studied for rule-based formalisms, providing reduction methods based on rule conditions that induce a lumping of the underlying Markov chain~\citep{Feret_IJSI2013,Feret2012137}.

For mass-action networks, the earlier approach to species lumping by~\citet{smbpaper}, called syntactic Markovian bisimulation, suffers from two limitations. First, syntactic Markovian bisimulation is only a sufficient condition for lumpability. Here, we prove that \se\ is the coarsest possible aggregation that yields a Markov chain lumping according to an equivalence over species. We show that this yields coarser aggregations than syntactic Markovian bisimulation in benchmark models.

The second limitation is that syntactic Markovian bisimulation only supports {\crn}s where reactions involve at most two reagents. Instead, SE can be applied to arbitrary higher-order reactions. This may appear unnecessary because in models of practical relevance reactions are typically of order two at most, following the basic principle that the probability of more than two bodies probabilistically colliding at the same time can be negligible~\citep{Gillespie77}. However, this generalization enables the identification of ``qualitative'' relations between species, i.e., equivalences that do not depend on the specific choice of values of the kinetic parameters. This is done by systematically turning the original network into one where each kinetic parameter appears as a further auxiliary species in a reaction, thus increasing its order by one. A number of case studies from the systems biology literature are used to show examples of parameter-independent physically intelligible model reductions.
\begin{figure*}[t]
\centering
\includegraphics{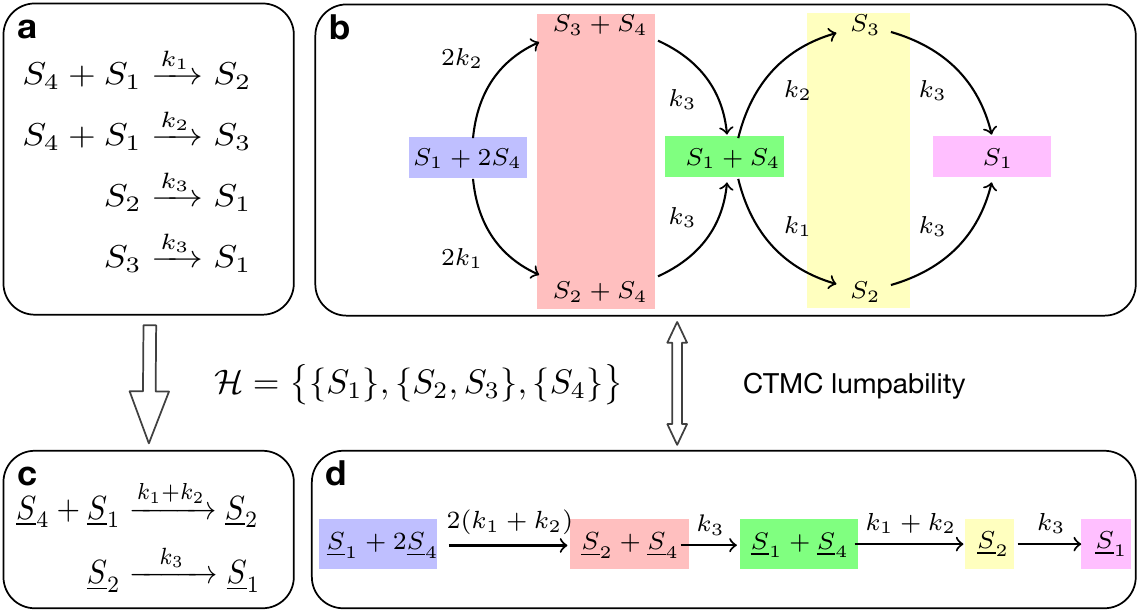}
%
\caption{Illustration of \sbe\ on a simple example.
a) Original \crn\ with four species $S_1$, \ldots, $S_4$.
b)~Underlying \ctmc\ derived from the initial state $\hat{\sigma} = S_1 + 2S_4$. The \ctmc\ is shown in the customary graph representation where each node is a state and transitions between two states $\sigma$, $\sigma'$ are directed arcs labeled with $q(\sigma, \sigma')$.  The colored boxes represent five blocks of an ordinary lumpable partition of the \ctmc\ (here it suffices to check that the outgoing transitions are equal for states in blocks of size two). The partition of species denoted by $\mathcal{H}$ can be shown to be an \sbe\ independently on the actual values of the parameters $k_1$, $k_2$, and $k_3$, hence states that are equal up to the sum of the second and third coordinate form a lumpable partition block.
c)~This \sbe\ gives rise to a reduced \crn\ by choosing the representatives $S_1$, $S_2$ and $S_4$ for each block (underlined in the figure to distinguish them from original species names). The reduced \crn\ has fewer reactions due to the fact that reactions in the original \crn\ are merged into a single one after renaming.
d)~Underlying \ctmc\ of the reduced \crn\ derived from the matching initial state $\underline{S}_1  + 2\underline{S}_4$. The \ctmc\ of the reduced network corresponds to the lumped \ctmc\ of the original network (as indicated by the matching colors of the nodes).
}\label{fig:sbe}
\end{figure*}

\section{Results}

\paragraph*{Stochastic mass-action networks}  Formally, a stochastic mass-action \crn\ is described by a set of species $\mathcal{S}$ 
and a set of reactions $\mathcal{R}$. 
Each reaction is in the form $\rho \xrightarrow{\ \alpha \ } \pi$,
where $\alpha > 0$ is a kinetic parameter and $\rho$ and $\pi$ are multisets of species called reagents and products, respectively. The multiplicity of species $S$ in $\rho$ is denoted by $\rho(S)$, which represents its stoichometry coefficient. The set of all reagents and products across all reactions in the \crn{} are denoted by $\Reag(\R)$ and $\Prod(\R)$, respectively.

A stochastic mass-action network gives rise to a continuous-time Markov chain (\ctmc) where each state $\sigma$ is a multiset of species. From a state $\sigma$ such that $\rho \subseteq \sigma$, a reaction $\rho \xrightarrow{\alpha} \pi$ induces a transition with mass-action propensity
$\alpha \prod_{S \in \rho}\binom{\sigma(S)}{\rho(S)}$
to state $\sigma + \pi - \rho$, where the plus and minus operators indicate multiset union and difference, respectively, while $S\in\rho$ denotes that $S$ belongs to the support of $\rho$, i.e. $\rho(S)>0$.
Given an initial state $\hat{\sigma}$, the state space can be derived by exhaustively applying the reactions to compute all possible states reachable from $\hat{\sigma}$.
We denote with $\Next(\sigma)$
the multiset of outgoing transitions  from state $\sigma$,
\begin{align}\nonumber
\Next(\sigma) \! = \! &
\left\{\!\!\!\left| \ \sigma \act{\lambda} \sigma + \pi - \rho \mid (\rho \!\act{\rate} \!\pi) \!\in\! \mathcal R, \lambda \!=\! \alpha\prod_{S\in\rho}\binom{\sigma(S)}{\rho(S)} \right|\!\!\!\right\}
\end{align}
For any two distinct states $\sigma$ and $\theta$, we denote by $q(\sigma,\theta)$ the sum of the propensities from $\sigma$ to $\theta$ across all reactions, that is
\[
q(\sigma,\theta) = \sum_{(\sigma\act{\lambda}\theta) \in \Next(\sigma)} \lambda .
\]
Moreover, we set $q(\sigma,\sigma)$ to be the negative sum of all possible transitions from state $\sigma$, i.e., $q(\sigma,\sigma) = - \sum_{\theta \neq \sigma} q(\sigma, \theta)$. These values form the CTMC generator matrix,
which characterizes the dynamical evolution of the \ctmc\ by means of the master equation $\dot{p} = p^T Q$. Each component of its solution, $p_\sigma(t)$, is the probability of being in state $\sigma$ at time $t$ starting from some initial probability distribution~\citep{kampen01}.

Figure~\ref{fig:sbe} shows a simple running example to summarize the main results of this paper using the network in Fig.~\ref{fig:sbe}a with species $S_1$, \ldots, $S_4$. The state space from the initial state $\hat{\sigma} = S_1 + 2S_4$ is in shown Fig.~\ref{fig:sbe}b.

Ordinary lumpability is a partition of the state space such that any two states $\sigma_1$, $\sigma_2$ in each partition block $H$ have equal aggregate rates toward states in any block $H'$, that is $\sum_{\sigma \in H'} q(\sigma_1,\sigma) = \sum_{\sigma \in H'} q(\sigma_2,\sigma)$~\citep{Ke76,BuchholzOrdinaryExact}. Given an ordinarily lumpable partition, a \emph{lumped} CTMC can be constructed by associating a macro-state to each block; transitions between macro-states are labelled with the overall rate from a state in the source block toward all states in the target. Distinct colored boxes in Fig.~\ref{fig:sbe}b identify an ordinarily lumpable partition of the sample \ctmc.
Ordinary lumpability preserves stochastic equivalence in the sense that the probability of each block/macro-state is equal to the sum of the probabilities in each original state belonging to that block.

\paragraph*{Species equivalence} Verifying the conditions for ordinary lumpability requires the full enumeration of the CTMC state space, which grows combinatorially with the multiplicities of initial state and the number of reactions. Additionally, the presence of interactions such as constitutive transcription, e.g., $S_1 \xrightarrow{\alpha} S_1 + S_2$, may give rise to infinite state spaces, preventing the use of ordinary lumpability altogether. SE detects ordinary lumpability at the finitary level of the reaction network by identifying an equivalence relation (i.e., a partition) of the species which induces an ordinary lumpable partition over the multisets representing CTMC states.

For this, we consider a natural lifting of a partition $\pa$ of species to multisets of species, which we call the \emph{multiset lifting} of $\pa$ and which we denote by $\upp{\pa}$.
It relates multisets that have same cumulative multiplicity from each partition block.
That is, two multisets/states $\sigma_1$ and  $\sigma_2$ belong to the same block  $\mb \in \upp{\pa}$ if the condition $\sum_{S \in H} \sigma_1(S) = \sum_{S \in H} \sigma_2(S)$ is satisfied for all blocks of species $H \in \pa$.

At the basis of \sbe\ is the notion of reaction rate $\rr(\rho,\pi)$ from reagents $\rho$ to products $\pi$,
\[
\rr(\rho,\pi) =
\begin{cases}
\sum\limits_{{(\rho \act{\rate} \pi) \in \mathcal{R}}} \rate &\quad \text{ if } \rho \neq\pi, \\
-\sum\limits_{\pi' \in \Prod(\mathcal{R}),  \rho\neq \pi'} \rr(\rho,\pi') &\quad \text{ if } \rho = \pi .
\end{cases}
\]
Intuitively, it is defined as the analogue to the entries of the \ctmc\ generator matrix, but it is computable by only inspecting the set of reactions. \sbe\ is defined as a partition of species $\pa$ such that, for any two species $S_i$ and $S_j$ in a block of $\pa$, and for any block of multisets $\mb \in \upp{\pa}$ containing at least one product in $\Prod(\R)$, the condition
\begin{equation}\label{eq:rr}
\sum_{\pi \in \mb} \rr(S_i + \rho, \pi) =  \sum_{\pi \in \mb} \rr(S_j + \rho, \pi),
\end{equation}
holds for all $\rho$  such that $S_i + \rho$ or $S_j + \rho$ are in the set of reagents $\Reag(\R)$.

According to this definition, species $S_2$ and $S_3$ in the sample network of Fig.~\ref{fig:sbe} belong to the same block of an \se. This explains why the ordinarily lumpable partition depicted in Fig.~\ref{fig:sbe}b groups \ctmc\ states that have the same total multiplicities of $S_2$ and $S_3$.

Our first result is that \sbe\ characterizes ordinary lumpability, in the sense that the multiset lifting of an \sbe\ yields an ordinarily lumpable partition of the underlying \ctmc\ derived from any initial state $\hat{\sigma}$; and, vice versa, if a multiset lifting of a partition of species $\pa$ is an ordinarily lumpable partition of the underlying \ctmc{} from any initial state $\hat{\sigma}$, then $\pa$ is an \sbe\ (proved in \sitext,Sec.~\ref{A.1}). We also note that, by~\citet{infiniteOL}, our result  also applies to {\ctmc}s with infinite state spaces, because each state has finitely many incoming and outgoing transitions due to the fact that the number of reactions is finite, and the state space is partitioned in blocks of finite size by multiset lifting.

\paragraph*{Computation of a reduced reaction network up to species equivalence} Analogously to the existence of a lumped \ctmc, one can build a reduced network from an \sbe\ partition. The reduction algorithm  is similar to that in~\citet{DBLP:conf/concur/CardelliTTV15}, where it was defined for deterministic mass-action networks with a reaction-rate interpretation based on ordinary differential equations. Briefly, the reduced \crn\ is obtained by applying the following four steps: (i) choose a representative species for each block of species; (ii) discard all reactions whose reagents have species that are not representatives; (iii) replace the species in the products of the remaining reactions with their representatives; (iv) reduce the set of reactions by merging all those that have same reactants and products by summing their kinetic parameters. The correctness of this algorithm is discussed in \sitext~(Sec.~\ref{A.2}). Following~\citet{DBLP:conf/concur/CardelliTTV15},  the reduced reaction network can be computed in $O( r s \log s)$ time, where $s$ is the number of species and $r$ is the number of  reactions.

Each representative in the reduced \crn\ can be interpreted as a macro-species that tracks the sum of the populations of the distinct species in the original \crn\ that belong to the same \sbe\ partition block. Therefore, for any given initial condition $\hat{\sigma}$ of the original \crn, it is possible to directly generate its lumped \ctmc\ from the reduced \crn\ by fixing a matching initial condition up to sums of populations, as related in general by multiset lifting. The network in Fig.~\ref{fig:sbe}c shows the reduced network up to an \se. The \ctmc{} obtained by ordinarily lumpability of the \ctmc{} in Fig.~\ref{fig:sbe}b corresponds to the \ctmc{} generated by the reduced network with the matching initial condition.

\paragraph*{Computation of the maximal \sbe} There exist efficient algorithms that compute the coarsest ordinarily lumpable partition, i.e., the maximal aggregation, of a \ctmc\ with a finite state space~\citep{DBLP:journals/ipl/DerisaviHS03,DBLP:conf/tacas/ValmariF10}. Here we develop an analogous algorithm for species of a reaction network. First, we show that, indeed, there exists the largest \se\ (\sitext{}, Sec.~\ref{A.3}).
Then, we develop a partition refinement algorithm that takes an initial partition of species as input and computes the largest \sbe\ that
refines such initial partition (\sitext{}, Sec.~\ref{A.4}). The maximal \sbe\ is thus a special case that can be computed by initializing the algorithm with the partition with the trivial singleton block containing all species.

The algorithm maintains a reference to the current candidate \sbe\ partition and a set of splitters, i.e., blocks of products against which the candidate partition is to be checked. Both structures are initialized using the input partition. A fixed-point iteration splits a block of the current candidate \sbe\ partition whenever it falsifies the condition in Eq.~\eqref{eq:rr} with respect to a splitter $\M$. If no such block is found, then the algorithm terminates and the candidate partition is proven to be the largest \sbe\ that refines the initial partition. Else, the falsifying block is split into sub-blocks that have equal values for the quantities in Eq.~\eqref{eq:rr}. The set of splitters is recomputed as the multiset lifting ot the current partition.
 We prove (in \sitext{}) that the algorithm has $O(p \ \! r)$ space and $O(s^2 r^3 p(p + \log r))$ time complexity, where $p$ is the largest number of different species appearing in the reagents or products of every reaction.

\paragraph*{Parameter-independent species equivalences by network expansion} Similarly to ordinary lumpability, checking the conditions of \se\ by Eq.~\ref{eq:rr} implicitly assumes that the values of all the kinetic parameters in the network are fixed. However, without further theory it is also possible to find equivalences that are independent from the specific values of the parameters. In order to do so, let $\mathcal{P}$ denote the set of all kinetic parameters used in the reaction network and assume, without loss of generality, that each kinetic parameter $\alpha \in \mathcal{P}$ is a rational number $n_\alpha/d_\alpha$. Let us then consider an expanded reaction network where we take each parameter $\alpha$ as an additional species $P_\alpha$, and every original reaction $\rho \act{\alpha} \pi$ is transformed into the reaction $P_\alpha + \rho \act{1} \pi + P_\alpha$. This is a reaction of higher order with kinetic parameter equal to one.

For this extended reaction network to be related to the original one,
each state of its CTMC must represent a multiset of species; in particular the initial condition of each additional species $P_\alpha$ must be a nonnegative integer, which will be fixed throughout the state space because the population of $P_\alpha$ does not change by construction. A suitable initialization of $P_\alpha$ may be for instance $\frac{n_\alpha}{d_\alpha} \text{lcm}\{ d_\alpha \mid \alpha \in \mathcal{P} \}$, where lcm denotes the least common multiple of all denominators of the parameters. With this in place, the original and the expanded network will give rise to the same state space (dropping the components of the CTMC state related to $P_\alpha$ because they are constant, as discussed). The transition rates of the expanded network are instead all scaled up by the same factor, which can essentially be interpreted as a time rescaling of the original CTMC.

Since this rescaling is the same for all states, any ordinarily lumpable partition on the CTMC of the expanded network will be an ordinary lumpable partition on the CTMC of the original one, and vice versa. More importantly, the computation of \se\ on the expanded network will be made independent of the specific values chosen for the kinetic parameters. This is because the parameter values are encoded into the components of the initial CTMC state associated with the auxiliary species $P_\alpha$, and \se\ finds equivalences that hold for all initial states of the \ctmc. Thus the computation of the largest \se\ in the original network may proceed by considering the initial partition consisting of two blocks, one for all the species and one for all species-parameters in the expanded network, respectively.
For the example in Fig.~\ref{fig:sbe}, the largest \se\ computed from the initial partition $\{ \{S_1, S_2, S_3, S_4\},  \{ P_{k_1}, P_{k_2}, P_{k_3} \} \}$ is $\{ \{S_1\}, \{S_2, S_3\}, \{ S_4\},  \{ P_{k_1}, P_{k_2}\}, \{ P_{k_3}\} \}$. In addition to the equivalence between the two species $S_2$ and $S_3$, it detects that the reduced model depends only on the sum of the parameters $k_1 + k_2$, for any given value.

\section{Examples}

In this section we present reductions on case studies from the literature, computed with an implementation of SE within the software tool ERODE~\citep{erode}, available at \href{https://www.erode.eu}{https://www.erode.eu}. The reported results refer to the analysis of the models with the values of the kinetic parameters as reported in the associated publications. However, the reductions are preserved also in the extended parameter-independent versions obtained as discussed above.

\paragraph*{Species equivalence in multi-site phosphorylation processes} Mechanistic models of signaling pathways are prone to a rapid growth in the number of species and reactions because of the combinatorial effects due to the distinct configurations in which a molecular complex can be found~\citep{FEBS:FEBS7027}. A prototypical situation is multisite phosphorylation, a fundamental process in eukaryotic cells that is responsible for various mechanisms such as the regulation of switch-like behavior~\citep{Gunawardena11102005,thomson2009}. For example, let us consider a protein $A$ with $n$ sites that can be phosphorylated by means of kinase $K$ according to a random mechanism, while dephosphorylation occurs as a spontaneous reaction. To describe this system one needs $2^n$ distinct molecular species that track the phosphorylation/dephosphorylation status of each site~\citep{FEBS:FEBS7027}. Each species is written in the form $A(s_1, \ldots, s_n)$ where $s_i = 0$ (resp., $s_i = 1$) indicates that the $i$-th site is dephosphorylated (resp., phosphorylated), for all $i = 1, \ldots, n$.
The resulting mass-action network is given by:
\begin{align*}
A(s_1, \ldots, s_{i-1}, 0, s_{i+1}, \ldots, s_n) + K & \xrightarrow{r_1} \\
& \!\!\!\!\!\!\!\!\! A(s_1, \ldots, s_{i-1}, 1, s_{i+1}, \ldots, s_n) , \\
A(s_1, \ldots, s_{i-1}, 1, s_{i+1}, \ldots, s_n) & \xrightarrow{r_2} \\
& \!\!\!\!\!\!\!\!\!\!\!\!\!\!\!\!\!\!\!\!\!\! A(s_1, \ldots, s_{i-1}, 0, s_{i+1}, \ldots, s_n) + K ,
\end{align*}
for all $i = 1, \ldots, n$ and for any combination of site states $s_1, \ldots, s_{i-1}, s_{i+1}, \ldots, s_n$. To simplify the mathematical model, it is assumed that the kinetic parameters $r_1$, $r_2$ are equal at all phosphorylation sites~\citep{citeulike:8493139}.

For a fixed $n$, the maximal \se\ aggregates molecular species that are equal up to the number of phosphorylated sites that they exhibit, independently of their identity. More formally, if we consider the  block of species $\Y_i$ that groups all configurations that have exactly $i$ phosphorylated sites,
$\Y_i  = \big\{ A(s_1, \ldots, s_n) \mid s_1 + \ldots + s_n = i  \big \}$, for $i = 0, \ldots, n$, then the maximal \se\ is given by the partition $\big\{ \{ K \}, \Y_0, \ldots, \Y_n \big \}$.

\paragraph*{Identification of equivalent molecular complexes in a model of synaptic plasticity} The assumption of equal kinetic parameters is not necessary to achieve aggregation with \sbe. We show this on a  model from~\citet{pepke2010dynamic} on the interactions between calcium (\calcium), calmodulin (\cam), and the \calcium-\cam\ dependent protein kinase II (\camkii), which play a fundamental role in the mechanism of synaptic plasticity~\citep{Lisman:2002aa}. (It is available in the BioModels database~\citep{BioModels2010}, identified as \textsl{MODEL1001150000}.)
The model describes the following processes: cooperative binding of \calcium\ to two pairs of domains located at the amino (N) and carboxyl (C) termini of \cam; binding of \cam\ to a monomeric \camkii\ subunit; and autophosphorylation of a \camkii\ monomer through the formation of a dimer which requires \cam\ to be bound to both subunits (Fig.~\ref{fig:calmodulin}A). The maximal \sbe\ finds that all phosphorylated monomers are equivalent (Fig.~\ref{fig:calmodulin}B), although their dynamics are characterized by distinct kinetic parameters to account for phosphorylation rates that depend on the number of bound \calcium~\citep{Shifman13968}. Further, such equivalences carry over to all complexes where they are present as sub-units. This leads to equivalence classes consisting of nine molecular species each, with an overall reduction from 155 species and 480 reactions to 75 species and 254 reactions.
Notably, important quantities to observe in this model are the amounts of free and bound \cam~\citep{Lisman:2012aa}, both recoverable from the reduced \crn.

\paragraph*{Internalization of the GTPase cycle in a model of the spindle position checkpoint} In both previous examples, \sbe\ can be physically interpreted as a reduction that preserves both the structure of equivalent molecular species as well as their function. \sbe\ can also aggregate
\begin{figure}[h!]
\centering
\includegraphics[width=0.95\linewidth]{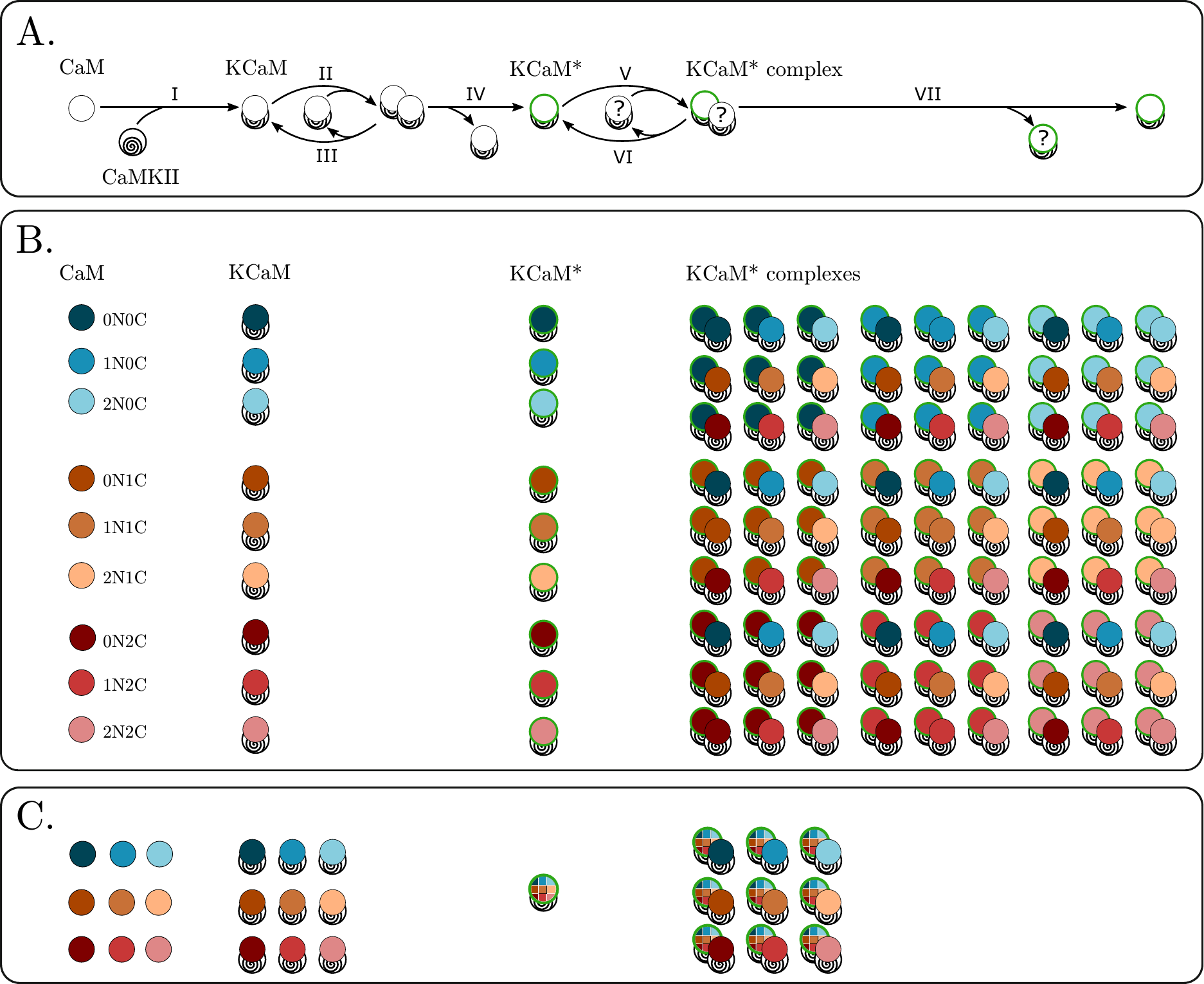}
\caption{(A)  Kinetic scheme of the interactions between \calcium-bound-\cam\ and CaMKII, adapted from~\citet{pepke2010dynamic}. \cam\ binds with the \camkii\ monomer to  form a KCaM complex (reaction I). KCaM may undergo reversible dimerazation (reactions II and III), KCaM dimerization which can lead to autophosphorylation (reaction IV). A unit of phosphorylated KCaM is labelled KCaM* and are represented in the diagram as a green-circled KCaM. KCaM* can interact with any unphosphorylated KCaM unit (indicated by the `?' sign) to form a KCaM* complex reversibly (reactions V and VI), leading to autophosphorlyation (reaction VII). (B) Molecular species that can participate in the reaction scheme. CaM units are represented with different colors to indicate the possible states of \calcium-binding, with the label xNyC, with $\text{x},\text{y} = 0,1,2$ denoting the number \calcium-bound domains at the amino and carboxyl termini, respectively. KCaM and KCaM* are represented similarly. The block of KCaM* complexes contains 81 distinct molecular species obtained from all possible interactions between KCaM* and KCaM. (C) The maximal SE yields a coarse-grained network which can be interpreted as having the same reaction scheme (A), but with fewer species. In particular CaM and KCaM complexes are not aggregated, but all the distinct KCaM* are collapsed into the same equivalence class (indicated by the multiple-color representative). Such equivalence carries over to all KCaM* complexes, in the sense that all dimers with the same phosphorylated form are in the same SE block. This allows the collapse of the distinct 81 KCaM* complexes above to 9 macro-species.}\label{fig:calmodulin}
\end{figure}
species that exhibit contrasting functionality, such as in signal transduction switches realized by GTP- and GDP-bound forms of GTPases. To show this, we consider the model in~\citet{caydasi2012dynamical} of the spindle position checkpoint (SPOC), a mechanism in the budding yeast responsible for detecting the correct alignment of the nucleus between mother and daughter cells~\citep{doi:10.1146/annurev.genet.37.042203.120656}.
(The BioModels identifier for this model is \textsl{BIOMD0000000699}.)
The most upstream event of the pathway involves GTPase Tem1, which is regulated by the GTPase-activating protein (GAP) complex composed of Bfa1 and Bub2. Under correct alignment the GAP complex is inhibited by a kinase Cdc5 phosphorylating Bfa1~\citep{gruneberg2000}; under misalignment, the kinase Kin4 phosphorylates Bfa1, preventing the inhibitory phosphorylation by Cdc5~\citep{Pereira:2005aa}.

\begin{figure}[t]
\includegraphics[width=\linewidth]{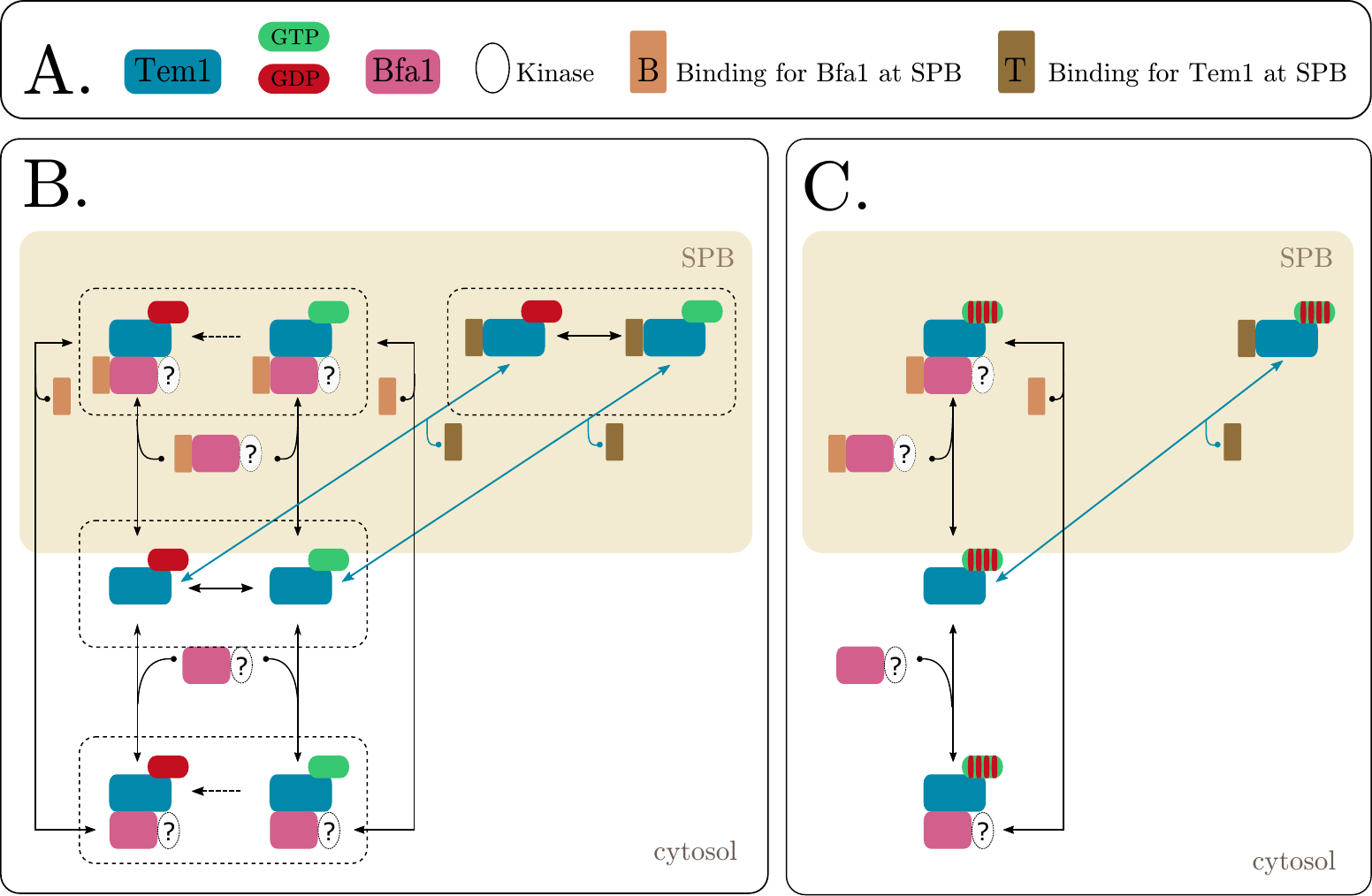}
\caption{Species equivalence for the SPOC dynamical model from~\citet{caydasi2012dynamical}. (A) Model subunits. (B) Illustration of the pathway. Beige boxes indicate the SPB compartment. Reactions crossing the compartment boundary represent the reversible SPB association of the respective species or complexes. Blue reactions mark the intrinsic Tem1 GTPase-cycle and reversible SPB association. Tem1 which is bound directly to the SPB does not interact with Bfa1, whilst Tem1 in the cytosol interacts with cytosol-Bfa1 and SPB-bound Bfa1. These interactions occur for all instances of Bfa1 regardless the Bfa1 phosphorylation state (indicated by the `?' symbol). GTP hydrolysis by the respective Bfa1-Tem1-GTP complexes (dashed reaction arrow) is accelerated according to the GAP activity of the respective state of Bfa1. The dashed boxes represent the SE equivalence classes indicating that the two forms of the GTPase Tem1 are equal up to \se. This equivalence extends to all complexes with the same configuration of subunits, up to GTP/GDP binding state of Tem1. (C) Graphical interpretation of the network reduced by \se.} \label{fig:spc}
\centering
\end{figure}

In the model, Tem1 binds to the yeast centrosomes (called spindle pole bodies, SPBs) via GAP-dependent and GAP-independent sites. The intrinsic GTPase switching cycle of Tem1 is modeled as a reversible first-order reaction that converts \ce{Tem1^{GTP}} into \ce{Tem1^{GDP}} and vice versa~\citep[Supplementary Material, Section 1]{caydasi2012dynamical}. The maximal \sbe\  collapses complexes that are  equal up to the GTP- or GDP-bound state, yielding eight equivalence classes with pairs of two molecular species (Fig.~\ref{fig:spc}). The original \crn\ with 24 species and 71 reactions is reduced to 16 species and 36 reactions, from which one may recover observables of interest such as the total amount of active Bfa1~\citep[Supplementary Material, Section 3]{caydasi2012dynamical}.

\paragraph*{Species equivalence for epidemic processes in networks} Models of epidemic processes are well established since the celebrated work by~\citet{sir1927}. 
The availability of large datasets in a range of socio-technical systems has prompted the study of epidemic processes on complex networks that consider the heterogeneity of real-world processes, which is neglected in simpler variants that assume a well-mixed, uniform environment~\citep{pastor2015epidemic}.

Aggregation of epidemic processes on networks has been studied in~\citet{simon2011exact}, relating symmetries in the graph with lumping on the \ctmc. Graph symmetry is formalized in terms of nodes belonging to the same orbit, thereby satisfying the property that there exists a graph automorphism relating them. Then, the orbit partition, i.e., the partition of nodes where each block is a distinct orbit, induces a \ctmc\ lumping that tracks the number of nodes in each block of the orbit partition that are in any given state~\citep{simon2011exact}.

Here we show that \sbe\ can be seen as a complementary, exact aggregation method for epidemic processes on complex networks. As an example, we study the well-known susceptible-infected-susceptible (SIS) model, where each node in the network in the susceptible state can be infected with a rate proportional to the number of infected neighbors, and recover from the infection according to an independent Poissonian process. Let $A = (a_{ij})$, with $A \in \mathbb{R}^{N \times N}$, define the adjacency matrix of a graph with $N$ nodes representing the network topology, with $a_{ij} > 0$ denoting the presence of a possibly weighted edge between node $i$ and $j$.

The SIS epidemic process can be described by the \crn\
\begin{align}\label{eq:sis}
S_i + I_j  &   \xrightarrow{\ a_{ij} \lambda\ } I_i + I_j,  &   I_i  &   \xrightarrow{\ \gamma \ } S_i, \ 1 \leq i,j \leq N, \ j \neq i,
\end{align}
where the first reaction models infections by neighbors and the second reaction is the spontaneous recovery, with parameters $\lambda$ and $\gamma$. 
In a similar fashion, different variants of the process, such as SIR, SIRS, and SEIR~\citep{pastor2015epidemic}, can be described. Any physically meaningful initial state $\hat{\sigma}$ for this \crn\ must be such that each node $i$ is initially in infected ($\hat{\sigma}_{S_i} = 0$,
$\hat{\sigma}_{I_i} = 1$) or susceptible ($\hat{\sigma}_{S_i} = 1$, $\hat{\sigma}_{I_i} = 0$).  This setting makes stochastic models of epidemics spreading on complex networks difficult to study exactly because the state of each individual node is  tracked explicitly~\citep{wang2017unification}, leading to a state space size with $2^N$ distinct configurations~\citep{simon2011exact}. \sbe\ provides an ordinary lumpability of the underlying \ctmc, without ever generating it, on the \crn\ of Eq.~\eqref{eq:sis}, which has exponentially smaller size because it has $2N$ species and $E + N$ reactions, where $E$ is the number of nonzero entries in the adjacency matrix of the graph.

\begin{figure}
\centering
\includegraphics[width=\linewidth]{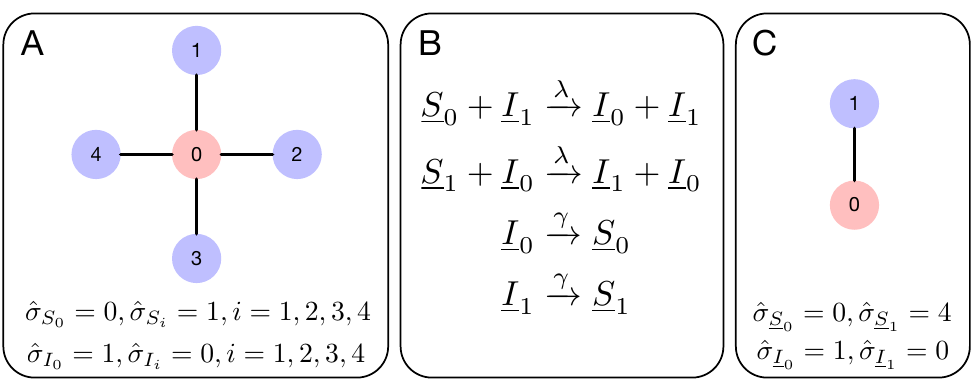}
\caption{Example of \sbe\ reduction of SIS dynamics on a coarse-grained network. (A)~Star network over which an SIS process evolves according to Eq.~\eqref{eq:sis}, starting from an initial condition where the infection starts at node $0$. (B)~Reduced network (species representatives are underlined in the figure for clarity) according to the largest \se\ refinement of the initial partition with blocks $\mathcal{S} = \{ S_0, S_1, S_2, S_3, S_4 \}$ and $\mathcal{I} = \{ I_0, I_1, I_2, I_3, I_4 \}$. This \sbe\ has blocks $\{ S_0 \}$, $\{ I_0 \}$, $\{ S_1, S_2, S_3, S_4 \}$ and $\{ I_1, I_2, I_3, I_4 \}$ . (C)~The \sbe\ partition induces a partition on the graph with blocks $\{0\}$ and $\{ 1,2,3,4\}$.  The reduced \crn\ corresponds to the description of the SIS dynamics on the quotient graph. The lumpability relation holds for an initial condition of the reduced \crn\ that is consistent with the initial condition of the original \crn\ up to \sbe.}\label{fig:sis}
\end{figure}
\begin{table*}[t]
\caption{%
Aggregation of SIS dynamics on benchmark networks. Models \textsf{tntp-ChicagoRegional}, \textsf{ego-facebook}, \textsf{as20000102}, \textsf{arenas-pgp}, \textsf{as-caida20071105}, \textsf{topology}, and \textsf{douban}, are taken from the Koblenz Network Collection~\citep{konect}; \textsf{web-webbase-2001} and \textsf{ia-email-EU} are taken from the Network Data Repository~\citep{nr}. $N = \text{number of vertices}$; $E = \text{number of edges}$ in the network.
}\label{tab:complex}
\centering
\scalebox{1.0}{
\begin{tabular}{llrrrrr}
\toprule
& & \multicolumn{2}{c}{\emph{Original size}} & \multicolumn{2}{c}{\emph{Reduced size}} \\
\cmidrule(l){3-4} \cmidrule(l){5-6}
\multicolumn{1}{c}{\emph{Network}}    &  Ref. & \multicolumn{1}{c}{$N$}  &  \multicolumn{1}{c}{$E$} &  \multicolumn{1}{c}{$N$}  & \multicolumn{1}{c}{$E$} &  \multicolumn{1}{c}{\emph{Orbits}} \\
\midrule
\textsf{tntp-ChicagoRegional}   &  \cite{konect:tntp-chicago1}  &        \np{1467} & \np{2596} & \np{635} & \np{932} & \np{166} \\
\textsf{ego-facebook}            & \cite{konect:McAuley2012}  & \np{2888} & \np{5962} & \np{35} & \np{104} & \np{35} \\
\textsf{as20000102}               & \cite{konect:leskovec107}  & \np{6474} & \np{27790} & \np{3885} & \np{19437} & \np{3690} \\
\textsf{arenas-pgp}               & \cite{konect:boguna}   & \np{10680} & \np{48632} & \np{8673} & \np{44074} & \np{7944} \\
\textsf{web-webbase-2001}      & \cite{boldi2004-ubicrawler} & \np{16062} & \np{51186} & \np{5253} & \np{24232} & \np{3574} \\
\textsf{as-caida20071105}     & \cite{konect:leskovec107}   & \np{26475} & \np{106762} & \np{13393} & \np{69184} & \np{13252} \\
\textsf{ia-email-EU}           & \cite{konect:leskovec107}  & \np{32430} & \np{108794} & \np{6262} & \np{53228} & \np{6259} \\
\textsf{topology}              & \cite{konect:zhang05}  & \np{34761} & \np{215440} & \np{19246} & \np{168782} & \np{19128} \\
\textsf{douban} & \cite{konect:socialcomputing} & \np{154908} & \np{654324} & \np{59524} & \np{462128} & \np{59493} \\
\bottomrule
\end{tabular}
}
\end{table*}

For the SIS model, the maximal \sbe\ is the trivial partition where all the species are in a single block. This is an invariant property stating that the total population of individuals in the system is constant~\citep{simon2011exact}. Thus, we consider non-degenerate reductions using initial partitions with two blocks,
$\{S_1, \ldots,  S_n\}$ and $\{I_1, \ldots, I_n\}$,
that separate species associated with nodes in the susceptible state from those in the infected state, respectively. As an illustrative example, let us consider the simple star graph~(Fig.~\ref{fig:sis}). An inspection of the obtained \sbe\ equivalence classes reveals that each refinement of the initial block $\{S_1, S_2, S_3, S_4\}$ matches a refinement of $\{I_1, I_2, I_3, I_4\}$ for the same subset of nodes of the graph. Such an  \sbe\ naturally induces a partitioning of the graph, and the reduction can be understood as an SIS dynamics on the quotient graph where each macro-node subsumes a partition block of nodes induced by \sbe.

We performed a systematic analysis of SIS processes evolving on several real-world benchmark networks (Table~\ref{tab:complex}), which confirms the observation made on the simple star graph. Since in all cases the reduced model is interpretable as an epidemic process, it is still amenable to a wide range of analysis techniques developed for such models~\citep{pastor2015epidemic,wang2017unification}. These include mean-field and pair approximation~\citep{VanMieghem2011,PhysRevE.85.056111,mata2013pair}, whose computational cost for the generation and solution of the resulting nonlinear differential equations may benefit from the availability of a stochastically equivalent reduced model.

Coarser aggregations of the \ctmc\ state space could be obtained in principle.
For example, the line graph in Fig.~\ref{fig:sis}C admits the orbit partition that collapses nodes 0 and~1, thereby inducing a lumping following~\citep{simon2011exact}. However, this is not detected by \sbe. Importantly, this does not contradict our characterization result. The reason is that the lumpability relation induced by \sbe\ must hold for all population vectors that are equal up to  \sbe. However, the lumpable partition derived with the approach in~\citet{simon2011exact} violates this property because it does not aggregate states $S_0 + I_0 + S_1 + I_1$ and $S_0 + S_0 + I_1 + I_1$, which preserve the sums of infected and susceptible individuals.
Indeed, in the real-world networks in Table~\ref{tab:complex} we found that \sbe\ always induces a partition on the nodes of the graph which is finer than the orbit partition (whose size is listed in the last column, as reported in~\citet{sym10010029}), albeit  not considerably so in some cases. On the other hand, \sbe\ can be applied to models that do not satisfy the conditions in~\citet{simon2011exact}. Indeed, the star network of Fig.~\ref{fig:sis} can be lumped also in the case of node-specific parameters (\sitext{}, Sec.~\ref{A.5}), while the results in~\citet{simon2011exact} require equal transmission and recovery rates at every node.

\paragraph*{Relationship with syntactic Markovian bisimulation} Applied to the models presented in this section, the earlier variant of SE, syntactic Markovian bisimulation~\citep{smbpaper}, yields the same reductions when applied to networks where the kinetic parameters are fixed. In \sitext{} (Sec.~\ref{A.6}) we present further models from the literature where SE yields maximal aggregations that are coarser than syntactic Markovian bisimulation, with up to about one order of magnitude fewer species.

\paragraph*{Speeding up stochastic simulations}
In \sitext{} (Sec.~\ref{A.7}), we use the same set of models to
 also provide evidence of the computational savings when analysing by stochastic simulation the reduced network in place of the original one. 
We report runtime speed-ups of up to three orders of magnitude using state-of-the-art algorithms as implemented in the StochKit simulation framework~\citep{DBLP:journals/bioinformatics/SanftWRFLP11}.

\section{Conclusion}

Stochasticity is a key tool to understand a variety of phenomena regarding the dynamics of reaction networks, but the capability of exactly analyzing complex models escapes us due to the lack of analytical solutions and the high computational cost of numerical simulations in general. Species equivalence enables aggregation in the sense of Markov chain lumping by identifying structural properties on the set of reactions, without the need of costly state-space enumeration. Owing to the polynomial space and time complexity of the reduction algorithm, it can be seen as a universal pre-processing step that exactly preserves the stochastic dynamics of species of interest to the modeler. Since it gives rise to a \crn\ where the reactions preserve the structure (up to a renaming of the species into equivalence classes), the reduction maintains a physical interpretation in terms of coarse-grained interactions between populations of macro-species. The possibility of computing  reductions that are not dependent from specific values of the kinetic parameters may reveal structural aggregations in the network, in addition to making the reduced model reusable across different parameter settings, e.g., when performing sensitivity analyses.

Being exact, our method can be combined with other techniques for the analysis of stochastic reaction networks. For instance, when feasible, one can generate the underlying \ctmc\ to be further analyzed or reduced~\citep{DBLP:conf/tacas/ValmariF10,doi:10.1063/1.2145882,henzinger2009sliding}; the reduced \crn\ can be subjected to complementary coarse-graining techniques concerned with time-scale separation (e.g.,~\citet{Sinitsyn10546,doi:10.1063/1.3050350,10.2307/23474859,doi:10.1063/1.4936394,cappelletti2016,BO20171}).
More generally, since the reduced \crn\ preserves the stochastic dynamics in the sense specified above, it can be used as the basis for other methods such as linear noise-or moment-closure approximation~\citep{1751-8121-50-9-093001}, where the complexity of the resulting system of equations depends on the network size.

\section*{Funding}

This work has been partially supported by Italian Ministry for Research under the PRIN project ``SEDUCE'', no. 2017TWRCNB, by the Independent Research Fund Denmark under the DFF RP1 Project REDUCTO no. 9040-00224B and the Danish Poul Due Jensen Foundation, grant
883901.


\bibliographystyle{natbib}
\bibliography{references}

\newpage
\onecolumn

\appendix[Supplementary Material]

\section{Supplementary Material}

\paragraph*{Notation} We begin by fixing preliminary notation. Given a set of species $\Y \subseteq \Sp$, we use $\rho(\Y) = \sum_{S \in \Y} \rho(S)$ to denote the cumulative multiplicity of all the species of $\Y$ in the multiset $\rho$.  
By definition of multiset lifting, all multisets in a block of $\upp{\pa}$ have same cumulative multiplicity from any block of $\Y\in\pa$; 
therefore, given a $\M\in\upp{\pa}$, we use $\M(\Y)$ to denote $\rho(\Y)$, where $\rho$ is any element of $\M$. Finally, we denote by $\ms(\Sp)$ the set of finite multisets of species in $\Sp$; given a multiset $\sigma$ and a set of multisets $G\subseteq\ms(\Sp)$ we use $q[\sigma,G]$ to denote $\sum_{\theta\in G}q(\sigma,\theta)$; similarly, we use $\rr[\rho,G]$ to denote the cumulative reaction rate $\sum_{\pi\in G}\rr(\rho,\pi)$.

\subsection{SE as a characterization for ordinary lumpability}
\label{A.1}

To prove the characterization result, we show the two directions separately. 
\paragraph*{\se\ is a sufficient condition} To prove that \se{} is a sufficient condition for ordinary lumpability, the first step is to express $q[\sigma, \tildeM]$, i.e., the cumulative transition rate from a \ctmc\ state $\sigma$ to states belonging to a block $\tildeM$ of the multiset lifting of an \se, in terms the reagents and products of the reactions that generate those transitions. Given source and target blocks $\M$ and $\tildeM$, respectively, and for a given block of multisets $\barM$, we define the set
\begin{align}\label{eq:blockHHH}
\Mgt{\barM}{\M}{\tildeM} =\{\pi \in\ms(\Sp) \mid
\exists \sigma\in\M, \rho\in\barM \text{ s.t. }
 \rho\subseteq\sigma \text{ and } (\sigma - \rho + \pi) \in \tildeM
\},
\end{align}
which collects all products $\pi$ of reactions which can be executed in a state $\sigma \in \M$, such that the reagents belong to $\barM$ and the target state is in block $\tildeM$. Importantly, it can be shown that $\Mgt{\barM}{\M}{\tildeM}$ is a block of the multiset lifting.
\begin{lemma}\label{rm:propertiesOfLifting}
Let $\Sp$ be a set, $\pa$ be a partition of $\Sp$, and $\eqpa{\pa}$ the equivalence  inducing it. Let $\upp{\pa}$ be the multiset lifting of $\pa$, and $\eqpa{\upp{\pa}}$ the equivalence on multisets inducing it. 
For all $\sigma,\sigma',\pi,\pi',\rho,\rho'\in\ms(\mathcal S)$, we have
\begin{enumerate}
\item $(\sigma\cup \pi,\sigma\cup \pi')\in \ \eqpa{\upp{\pa}}$ if and only if $(\pi,\pi')\in \ \eqpa{\upp{\pa}}$,
\item if $(\sigma,\sigma')\in \ \eqpa{\upp{\pa}}$, then $(\sigma\cup \pi,\sigma'\cup \pi')\in \ \eqpa{\upp{\pa}}$ if and only if $(\pi,\pi')\in \ \eqpa{\upp{\pa}}$,
\item $(\sigma- \rho,\sigma- \rho')\in \ \eqpa{\upp{\pa}}$ if and only if $\rho\subseteq\sigma\supseteq\rho'$ and $(\rho,\rho')\in \ \eqpa{\upp{\pa}}$,
\item if $(\sigma,\sigma')\in \ \eqpa{\upp{\pa}}$, then $(\sigma - \rho,\sigma'- \rho')\in \ \eqpa{\upp{\pa}}$ if and only if $\rho\subseteq\sigma$, $\rho'\subseteq\sigma'$ and $(\rho,\rho')\in \ \eqpa{\upp{\pa}}$.
\end{enumerate}
For any $\M,\tildeM \in \upp{\pa}$, if it is possible to obtain multisets in $\tildeM$ by adding species to those in $\M$, i.e., if $\M(\Y)\leq\tildeM(\Y)$ for all $\Y\in\pa$, then
\begin{itemize}
\item From points $1$, $2$ we have that 
there exists one $\hatM\in\upp{\pa}$ such that  $\tildeM=\{\sigma+\hat\sigma\mid \sigma\in\M, \hat\sigma\in\hatM\}$. That is, we obtain $\tildeM$ by pairwise merging the multisets in $\M$ with those in $\hatM$.
\item From points $3$, $4$ we have that the set $\{\tilde\sigma - \sigma \mid \tilde\sigma\in\tilde\M, \sigma\in\M, \sigma\subseteq \tilde\sigma\}$ is a block of $\upp{\pa}$.
\end{itemize}
\end{lemma}
We omit the proof of this lemma as it is straightforward.

With this, we can prove that for any two distinct blocks of the multiset lifting, $\M,\tildeM\in \upp{\pa}$ and for any $\sigma \in \M$ it holds that
        \begin{equation}\label{eq:p1}
        q[\sigma,\tildeM]= \sum_{\barM\in\upp{\pa}}\sum_{\stackrel{\rho\in\barM}{\rho \subseteq \sigma}}\prod_{S \in \rho}\binom{\sigma(S)}{\rho(S)}\cdot\rr[\rho,\Mgt{\barM}{\M}{\tildeM}]
        \ ,
        \end{equation}
thus expressing the aggregate rate in terms of the source state $\sigma$ and quantities depending on the multiset lifting. This is formally stated as follows.

\begin{proposition}\label{prop:propertiesOfLifting}
Let $(\Sp,\R)$ be a reaction network, $\pa$ a partition of $\mathcal S$, and $\upp{\pa}$ its multiset lifting.
Further, let $\M,\tildeM\in \upp{\pa}$ such that $\M \neq \tildeM$. Then, for any $\sigma \in \M$ it holds that
\[q[\sigma,\tildeM]= \sum_{\barM\in\upp{\pa}}\sum_{\stackrel{\rho\in\barM}{\rho \subseteq \sigma}}\prod_{S \in \rho}\binom{\sigma(S)}{\rho(S)}\cdot\rr[\rho,\Mgt{\barM}{\M}{\tildeM}]
\ .\]
\end{proposition}

\begin{proof}
By  definition we have
\begingroup
\allowdisplaybreaks
\begin{align}
q[\sigma,\tildeM]
&= \sum_{\theta \in \tildeM}\sum_{(\sigma \act{\lambda} \theta) \in \Next(\sigma)} \lambda \notag
\\
&= \sum_{\theta \in \tildeM}\sum_{\stackrel{(\rho \act{\rate} \pi) \in \R}{(\sigma -\rho+ \pi)=\theta}}  \prod_{S \in \rho}\binom{\sigma(S)}{\rho(S)}\cdot \rate \notag \\
& = \sum_{\stackrel{(\rho \act{\rate} \pi) \in \R}{(\sigma -\rho+ \pi)\in\tildeM}} \prod_{S \in \rho}\binom{\sigma(S)}{\rho(S)}\cdot \rate \notag
\end{align}
We now explicitly sum over all possible reagents $\rho$ in $\ms(\Sp)$, partitioning them according to $\upp{\pa}$. In addition, we restrict to those $\rho$ contained in $\sigma$, as they are the only ones actually considered in the above summation. This gives
\[
q[\sigma,\tildeM] = \sum_{\barM\in\upp{\pa}}\sum_{\stackrel{\rho\in\barM}{\rho\subseteq\sigma}}\sum_{\stackrel{(\rho \act{\rate} \pi) \in \R}{(\sigma -\rho+ \pi)\in\tildeM}} \prod_{S \in \rho}\binom{\sigma(S)}{\rho(S)}\cdot \rate
\]
The product above does not depend on each reaction considered in the innermost summation, but on the reagents $\rho$. Therefore we can factor out the product, obtaining
\begin{align*}
q[\sigma,\tildeM] = \sum_{\barM\in\upp{\pa}}\sum_{\stackrel{\rho\in\barM}{\rho\subseteq\sigma}}\prod_{S \in \rho}\binom{\sigma(S)}{\rho(S)} \!\!\!\!\!\!\! \sum_{\stackrel{(\rho \act{\rate} \pi) \in \R}{(\sigma -\rho+ \pi)\in\tildeM}}  \rate
\end{align*}
By using Lemma~\ref{rm:propertiesOfLifting}, this can be rewritten as:
\begin{align*}
q[\sigma,\tildeM] & = \sum_{\barM\in\upp{\pa}}\sum_{\stackrel{\rho\in\barM}{\rho \subseteq \sigma}}\prod_{S \in \rho}\binom{\sigma(S)}{\rho(S)}\sum_{\pi\in\Mgt{\barM}{\M}{\tildeM}}\sum_{(\rho \act{\rate} \pi) \in \R} \rate
\\
&= \sum_{\barM\in\upp{\pa}}\sum_{\stackrel{\rho\in\barM}{\rho \subseteq \sigma}}\prod_{S \in \rho}\binom{\sigma(S)}{\rho(S)}\cdot\rr[\rho,\Mgt{\barM}{\M}{\tildeM}] \notag
\end{align*}
\endgroup
\end{proof}

The next auxiliary fact is a purely combinatorial result whose role is to express $q[\sigma,\tildeM]$ in a way that does not depend on the source state $\sigma$, but only on the block of the multiset lifting to which it belongs.
\begin{proposition}\label{BinomEquality}
Let $\Sp$  be a set, $\pa$ a partition of $\Sp$, and $\upp{\pa}$ its multiset lifting.
For any $\M, \tildeM \in \upp{\pa}$, for any $\sigma \in \M$, it holds that
\[
\sum_{\rho\in\tildeM}\prod_{S \in \rho}\binom{\sigma(S)}{\rho(S)}  = \prod_{\Y \in \pa} \binom{\M(\Y)}{\tildeM (\Y)}\label{eq:BinomEq}
\]
\end{proposition}

\begin{proof}
We can rewrite the left-hand side of the above equation as follows:
\begin{align}
	\sum_{\rho\in\tildeM}\prod_{\Y \in \pa}\prod_{\stackrel{S \in \rho}{S\in \Y}}\binom{\sigma(S)}{\rho(S)} \nonumber
\end{align}
which by using $(\tildeM)_{ | \Y}$ to denote the set obtained by projecting each $\rho\in\tildeM$ to $\Y$, can be further rewritten as
\[
\prod_{\Y \in \pa} \sum_{\rho' \in (\tildeM)_{ | \Y}} \prod_{X \in \Y} \binom{\sigma(X)}{\rho'(X)}
\]
We note that taking elements $\rho'\in(\tildeM)_{ | \Y}$ instead of taking each $\rho\in\tildeM$ and projecting it to $\Y$  makes sure that each distinct multiset $\rho'$ is counted only once in the above equation.

Now, for any $\rho'\in(\tildeM)_{ | \Y}$ we have that $\sum_{S\in\Y}\rho'(S)=\rho'(\Y) = (\tildeM)_{| \Y}(\Y)= \tildeM(\Y)$ is constant for each $\Y\in\pa$, just as $\sum_{X\in\Y}\sigma(X)=\sigma(\Y) = \M(\Y)$ for all $\sigma \in \M$. This allows us to apply Vandermonde's identity, obtaining:
\begin{align*}
\prod_{\Y \in \pa} \sum_{\rho' \in (\tildeM)_{ | \Y}} \prod_{S \in \Y} \binom{\sigma(S)}{\rho'(S)}  &=
 \prod_{\Y \in \pa} \binom{\sum_{S\in\Y}\sigma(S)}{\sum_{S\in\Y}\rho'(S)}
= \prod_{\Y \in \pa} \binom{\sigma(\Y)}{\rho'(\Y)}
= \prod_{\Y \in \pa} \binom{\M(\Y)}{\tildeM (\Y)}
\end{align*}
This closes the proof.
\end{proof}

Then, to formally prove that \se\ is a sufficient condition for ordinary lumpability we need to show that for any two blocks of the multiset lifting $\M,\tildeM \in \upp{\pa}$ and for any two states $\sigma,\sigma' \in \M$, we have that
\[q[\sigma,\tildeM]=q[\sigma',\tildeM].\]

Indeed, let us assume that $\M \not = \tildeM$. Due to the properties of SE, we can factor out the reaction-rate quantities in Eq.~\eqref{eq:p1}. Using $\rho^{\barM}$ to denote any element of $\barM$, we obtain
\begin{align}
q[\sigma,\tildeM]
& = \sum_{\barM\in\upp{\pa}}\rr[\rho^{\barM},\Mgt{\barM}{\M}{\tildeM}]\sum_{\stackrel{\rho\in\barM}{\rho \subseteq \sigma}}\prod_{S \in \rho}\binom{\sigma(S)}{\rho(S)} \notag\\
& =\sum_{\barM\in\upp{\pa}}\rr[\rho^{\barM},\Mgt{\barM}{\M}{\tildeM}]\sum_{\rho\in\barM}\prod_{S \in \rho}\binom{\sigma(S)}{\rho(S)}  \notag\\
& = \sum_{\barM\in\upp{\pa}}\rr[\rho^{\barM},\Mgt{\barM}{\M}{\tildeM}] \prod_{\Y \in \pa} \binom{\M(\Y)}{\barM(\Y)} , \label{eq:suffsigma} 
\end{align}
where the last equality follows by Proposition~\ref{BinomEquality}. 
This closes the case, because the terms appearing on the right-hand side of Eq.~\eqref{eq:suffsigma} do not depend on $\sigma$ or $\sigma'$. Since the case $\M = \tildeM$ follows from the case $\M \not = \tildeM$, see~\cite[Proposition 1]{DBLP:conf/tacas/ValmariF10}, we infer the sufficiency of SE.


\paragraph*{\se\ is a necessary condition for ordinary lumpability} Our necessary condition states that if a multiset lifting $\upp{\pa}$ is an ordinary lumpable partition of the underlying \ctmc{} for any initial state $\hat{\sigma}$, then $\pa$ is an \sbe. To prove this statement for any initial state, it is convenient to consider the \ctmc\ where the state space is the whole set of finite multisets of species $\ms(\Sp)$, since this will include the state space generated by any initial condition. We will thus denote by $\MC(\Sp,\R)$ the \ctmc\ of the reaction network with species $\Sp$ and reactions $\R$ on the whole state space. Moreover, for any $m \geq 1$ and $\bowtie\ \in \{ =, \leq \}$, let $\R_{\bowtie m}$ denote the subset of reactions of $\R$ where the multiplicity of the reagents is at most $m$, i.e., $\R_{\bowtie m} = \{ (\rho \act{\alpha} \pi) \in \R \mid  |\rho| \bowtie m\}$. We denote by  $q_{\bowtie m}(\sigma,\theta)$ the \ctmc\ transition rate from $\sigma$ into $\theta$ in $\MC(\Sp,\R_{\bowtie m})$. 

Then, to prove the necessary condition we need to show that for all states $\rho$, $\rho' \in \MS(\Sp)$ and for all blocks  $\M,\tilde{\M} \in\upp{\pa}$ we have that
$\rho,\rho' \in \M \text{ implies } \rr[\rho,\tildeM] = \rr[\rho',\tildeM]$.

To this end, we proceed by induction on $|\rho| = m$:
\begin{itemize}
        \item $|\rho| = 1$: Then $\rr[\rho,\tildeM] = q[\rho,\tildeM] = q[\rho',\tildeM] = \rr[\rho',\tildeM]$, where the second identity follows from the assumption.
        \item $|\rho| = m + 1$: Thanks to induction hypothesis, $\pa$ is an SE of $(\Sp,\R_{\leq m})$. By applying the result of sufficient condition to $(\Sp,\R_{\leq m})$ and $\pa$, we infer that $\upp{\pa}$ is an ordinary lumpability of $\MC(\Sp,\R_{\leq m})$. Since this implies that $q_{\leq m}[\rho,\tildeM] = q_{\leq m}[\rho',\tildeM]$, the identities
            \begin{align*}
            q[\rho,\tildeM] & = q_{=m+1}[\rho,\tildeM] + q_{\leq m}[\rho,\tildeM] &
            \text{and} &  &
            q[\rho',\tildeM] & = q_{=m+1}[\rho',\tildeM] + q_{\leq m}[\rho',\tildeM]
            \end{align*}
            ensure $q_{=m+1}[\rho,\tildeM] = q_{=m+1}[\rho',\tildeM]$ because $q[\rho,\tildeM] = q[\rho',\tildeM]$ by assumption. This, in turn, yields
        \begin{align*}
        \rr[\rho,\tildeM] = q_{=m+1}[\rho,\tildeM] = q_{=m+1}[\rho',\tildeM] = \rr[\rho',\tildeM],
        \end{align*}
\end{itemize}
which completes our proof.

\subsection{Reduced reaction network up to species equivalence}
\label{A.2}

Given a partition $\pa$ of $\cal S$, a block $\Y\in\pa$ and a species $S\in\Y$, we use $\reduce{S}{\pa}$ to denote the canonical representative of the block $\Y$. For a multiset of species $\rho$, we set $\reduce{\rho}{\pa} = \sum_{S \in \rho} \reduce{S}{\pa}$ as the \emph{representative multiset} obtained replacing each species with its canonical representative.
Similarly, for any set of multi-sets $G\subseteq \MS(\Sp)$ we use $\reduce{G}{\pa} = \{\reduce{\rho}{\pa}\mid \rho\in G\}$ to denote the set of $\pa$-reduced multisets of $G$.
Any block of multisets $\M$ in the lifting $\upp{\pa}$ has a unique $\pa$-reduced representative which we denote by $\reduce{\M}{\pa}$. We use $\reduce{\mathcal R}{\pa}$ to denote the set of reactions in the reduced reaction network.

If $\pa$ is an \se, by definition the cumulative value of $\rr$ from a multiset $\rho$ to a block $\M$ of the multiset lifting does not change if species $S$ in $\rho = S + \rho'$ is replaced by its representative $\reduce{S}{\pa}$ from the same partition $\Y \in \pa$, that is:
\begin{align*}
\rr[S+\rho',\M]=
\rr[\reduce{S}{\pa}+\rho',\M].
\end{align*}
Then, by induction it follows that
\begin{align*}
 \rr[\rho,\M]=
\rr[\reduce{\rho}{\pa},\M]
\end{align*}
Furthermore, in the case $\rho\not\in\M$, by definition of reaction rate we have that
\[
\rr[\reduce{\rho}{\pa},\M] = \sum_{\pi \in \M}\sum_{(\reduce{\rho}{\pa}\act{\alpha}\pi)\in \mathcal{R}}\alpha
\]
Note that the inner summation in the above equation is  across reactions whose reagents contain only representative species.
The outer summation sums across products that have the same multiset representative by construction. Hence, by the definition of reduced \rn, overall we have
\begin{equation}\label{old:prop:rrreduced}
\rr[\rho,\M] = \rr[\reduce{\rho}{\pa},\M] = \reduce{\rr}{\pa}(\reduce{\rho}{\pa}, \reduce{\M}{\pa}
)
\end{equation}
where by $\reduce{\rr}{\pa}$ we denote the computation occurring in the reduced reaction network.

By using Equation~\eqref{old:prop:rrreduced} as an intermediate step, we now show that the reduced network yields the lumped \ctmc\ by  proving that the aggregate rate from any state of the original \ctmc\ toward any block of the ordinarily lumpable partition corresponds to the single transition rate from their respective representative multisets, formally
\begin{align}\label{eq:claimRed}
q[\sigma,\tildeM]=
\reduce{q}{\pa}(\reduce{\sigma}{\pa},\reduce{\tildeM}{\pa})
\quad \text{ for any }\tildeM \in \upp{\pa}\ ,
\end{align}
where $\reduce{q}{\pa}$ indicates the values computed for the reduced reaction network. To see this, we proceed by case distinction on $\tildeM\in\upp{\pa}$ such that either $\sigma \in \tildeM$ or  $\sigma \not\in \tildeM$.
We start with the latter case, for which we also have  $\reduce{\sigma}{\pa}\not=\reduce{\tildeM}{\pa}$. By using Proposition~\ref{prop:propertiesOfLifting}, we know that
\begin{align}\label{eq:qAsrr}
q[\sigma,\tildeM] = \sum_{\barM\in\upp{\pa}}\sum_{\stackrel{\rho\in\barM}{\rho \subseteq \sigma}}\prod_{S \in \rho}\binom{\sigma(S)}{\rho(S)}\cdot\rr[\rho,\Mgt{\barM}{\M}{\tildeM}]
\end{align}

By definition of SE, using $\rho^{\barM}$ to denote any $\rho\in\barM$, we can rewrite Eq.~\eqref{eq:qAsrr} as
\begin{equation} \label{eq:reducedModelHasLumpedCTMC}
q[\sigma,\tildeM] = \sum_{\barM\in\upp{\pa}}\rr[\rho^{\barM},\Mgt{\barM}{\M}{\tildeM}]\sum_{\stackrel{\rho\in\barM}{\rho \subseteq \sigma}}\prod_{S \in \rho}\binom{\sigma(S)}{\rho(S)}
\end{equation}

Considering instead $\reduce{q}{\pa}(\reduce{\sigma}{\pa},\reduce{\tildeM}{\pa})$, 
by definition we have 
\begin{equation}\label{eq:redCTMCX}
\reduce{q}{\pa}(\reduce{\sigma}{\pa},\reduce{\tildeM}{\pa})
 = \sum_{(\reduce{\sigma}{\pa} \act{\lambda} \reduce{\tildeM}{\pa}) \in \Next(\reduce{\sigma}{\pa})} \lambda
 = \sum_{\stackrel{(\rho \act{\rate} \pi) \in \reduce{\mathcal R}{\pa}}{(\reduce{\sigma}{\pa} -\rho+ \pi)=\reduce{\tildeM}{\pa}}} \!\! \prod_{S \in \rho}\binom{\reduce{\sigma}{\pa}(S)}{\rho(S)}\cdot \rate
\end{equation}

We obtain our claim from Eq.~\eqref{eq:claimRed} by showing that Eq.~\eqref{eq:reducedModelHasLumpedCTMC} is equal to Eq.~\eqref{eq:redCTMCX}. 
Indeed, Eq.~\eqref{eq:redCTMCX} can be rewritten as
\[
\reduce{q}{\pa}(\reduce{\sigma}{\pa},\reduce{\tildeM}{\pa}) = \sum_{\barM\in\upp{\pa}}\sum_{\stackrel{\rho\in\barM}{\rho\subseteq\reduce{\sigma}{\pa}}} \prod_{S \in \rho}\binom{\reduce{\sigma}{\pa}(S)}{\rho(S)}
  \sum_{\stackrel{(\rho \act{\rate} \pi) \in \reduce{\mathcal R}{\pa}}{(\reduce{\sigma}{\pa} -\rho+ \pi)=\reduce{\tildeM}{\pa}}} \!\!\!\!\!\!\!\! \rate.
\]
If $\rho\subseteq \reduce{\sigma}{\pa}$, then $\rho=\reduce{\rho}{\R}$, therefore it holds that 
\begin{align}
\reduce{q}{\pa}(\reduce{\sigma}{\pa},\reduce{\tildeM}{\pa})
& = \sum_{\barM\in\upp{\pa}}\sum_{\stackrel{\rho\in\barM}{\rho\subseteq\reduce{\sigma}{\pa}}} \prod_{S \in \reduce{\rho}{\pa}}\binom{\reduce{\sigma}{\pa}(S)}{\reduce{\rho}{\pa}(S)}
  \sum_{\stackrel{(\reduce{\rho}{\pa} \act{\rate} \pi) \in \reduce{\mathcal R}{\pa}}{(\reduce{\sigma}{\pa} -\reduce{\rho}{\pa}+ \pi)=\reduce{\tildeM}{\pa}}} \!\!\!\!\!\!\!\! \rate \notag\\
& = \sum_{\barM\in\upp{\pa}}\sum_{\stackrel{\rho\in\barM}{\rho\subseteq\reduce{\sigma}{\pa}}}
\prod_{S\in \reduce{\rho}{\pa}}\binom{\reduce{\sigma}{\pa}(S)}{\reduce{\rho}{\pa}(S)}
  \cdot\reduce{\rr}{\pa}(\reduce{\rho}{\pa},\hat\pi),\label{eq:reduceQIntermediateStep}
\end{align}
 where in the last equality we used the fact that for any
$\reduce{\rho}{\pa}\subseteq\reduce{\sigma}{\pa}$
there exists one $\hat\pi\in\ms(\mathcal S)$ such that $(\reduce{\sigma}{\pa}- \reduce{\rho}{\pa} + \hat\pi) = \reduce{\tildeM}{\pa}$
%
(such $\hat\pi$ might not exist, in which case the reaction rate is equal to $0$). 
Furthermore, we know that $\hat\pi$ belongs to $\Mgt{\barM}{\M}{\tildeM}$, and, in particular, that $\reduce{(\Mgt{\barM}{\M}{\tildeM})}{\pa}=\hat\pi$.

Now we can use Equation~\eqref{old:prop:rrreduced} to further rewrite Eq.~\eqref{eq:reduceQIntermediateStep} as follows
\[
\reduce{q}{\pa}(\reduce{\sigma}{\pa},\reduce{\tildeM}{\pa}) = \sum_{\barM\in\upp{\pa}}\sum_{\stackrel{\rho\in\barM}{\rho\subseteq\reduce{\sigma}{\pa}}} \prod_{S \in \reduce{\rho}{\pa}}\binom{\reduce{\sigma}{\pa}(S)}{\reduce{\rho}{\pa}(S)}
   \cdot
   \rr[\reduce{\rho}{\pa},\Mgt{\barM}{\M}{\tildeM}]
\]
which, by the definition of reduced network 
gives
\[
\reduce{q}{\pa}(\reduce{\sigma}{\pa},\reduce{\tildeM}{\pa}) = \sum_{\barM\in\upp{\pa}}\sum_{\stackrel{\rho\in\barM}{\rho\subseteq\sigma}} \prod_{S \in \rho}\binom{\sigma(S)}{\rho(S)}
   \cdot
   \rr[\rho^{\barM},\Mgt{\barM}{\M}{\tildeM}]
\]
Finally, by the properties of SE we can factor out the $\rr$ after the first summation, obtaining
\begin{align}\label{eq:redCTMCx-4}
\reduce{q}{\pa}(\reduce{\sigma}{\pa},\reduce{\tildeM}{\pa}) = \sum_{\barM\in\upp{\pa}}
\rr[\rho^{\barM},\Mgt{\barM}{\M}{\tildeM}]\cdot
\sum_{\stackrel{\rho\in\barM}{\rho\subseteq\sigma}} \prod_{S \in \rho}\binom{\sigma(S)}{\rho(S)}
\end{align}

This completes the case $\sigma \notin \tildeM$, as the $\Mgt{\barM}{\M}{\tildeM}$ considered in Eq.~\eqref{eq:redCTMCx-4} correspond to the $\Mgt{\barM}{\M}{\tildeM}$ considered in Eq.~\eqref{eq:reducedModelHasLumpedCTMC}.

We now present the case  $\sigma \in \tildeM$. First, we show that
the following equation holds:
\begin{align}\label{eq:qselfloop}
q[\sigma,\tildeM] = - \sum_{\hatM \in \upp{\pa},\ \tildeM \not= \hatM} q[\sigma,\hatM]
\end{align}
In particular, given that $\sigma \in \tildeM$, 
we have
\begin{align*}
q[\sigma,\tildeM] = q[\sigma,\tildeM\setminus \{\sigma\}] + q(\sigma,\sigma)\;.
\end{align*}
By definition, 
we have that
$q(\sigma,\sigma) = - q[\sigma,\MS(\mathcal S) \setminus \{\sigma\}]$.
If we partition $\ms(\mathcal S)$ according to $\upp{\pa}$, we obtain
\[
q(\sigma,\sigma) = - q[\sigma,\tildeM \setminus \{\sigma\}]  - \sum_{\hatM \in \upp{\pa},\ \tildeM \not= \hatM} q[\sigma,\hatM]
\]
which proves that Eq.~\eqref{eq:qselfloop} holds.
Now, from the case $\sigma \notin \tildeM$, it follows that
\[
- \sum_{\hatM \in \upp{\pa},\  \tildeM \not= \hatM} q[\sigma,\hatM]
= - \sum_{\hatM \in \upp{\pa},\  \tildeM \not=  \hatM} \reduce{q}{\pa}[\reduce{\sigma}{\pa}, \reduce{\hatM}{\pa}]
\]
Considering that $\sigma \in\tildeM$ implies that $\reduce{\sigma}{\pa} = \reduce{\tildeM}{\pa}$, this resolves to
\[
- \sum_{\stackrel{\reduce{\hatM}{\pa} \in \MS(\reduce{\mathcal S}{\pa})}{ \reduce{\sigma}{\pa} \not=  \reduce{\hatM}{\pa}}} \reduce{q}{\pa}(\reduce{\sigma}{\pa}, \reduce{\hatM}{\pa})
\]
This is by definition 
exactly $\reduce{q}{\pa}(\reduce{\sigma}{\pa}, \reduce{\sigma}{\pa})$.

\begin{figure}[t]
\centering
\begin{minipage}[c]{0.49\textwidth}
\centering
\begin{algorithmic}
\REQUIRE $\mathcal{R}$ reactions
\REQUIRE $\Reag(\mathcal{R})$ reagents
\ENSURE $\mathcal{R'}$ pre-processed reactions
\STATE $\mathcal{R}' \leftarrow \emptyset$
\FOR{$\rho \in \Reag(\mathcal R)$}
\STATE $tot \leftarrow 0$
\FOR{$(\rho \act{\rate}\pi) \in \mathcal R$}
\STATE $\text{add}((\rho\act{\rate}\pi),\mathcal{R}')$
\STATE $tot \leftarrow tot + \rate$
\ENDFOR
\STATE $\text{add}((\rho\act{-tot}\rho),\mathcal{R}')$
\ENDFOR
\end{algorithmic}
\vspace{4.95cm}
\end{minipage}
\hfill
\begin{minipage}[c]{0.49\textwidth}
\centering
\begin{algorithmic}
\REQUIRE $\mathcal{S}, \mathcal{R}$ reaction network
\REQUIRE $\mathcal{H}_0$ initial partition of $\mathcal{S}$
\ENSURE $\mathcal{H}$ the largest \se\ which refines $\mathcal{H}_0$
\STATE $\mathcal{H} \leftarrow \mathcal{H}_0$
\STATE $\A(\mathcal{R}) := \{\rho' \mid \exists S \in \mathcal{S} \text{~s.t.~} \exists(S + \rho') \in \Reag(\mathcal{R})\}$
\STATE $\upp{\pa}(\mathcal{R}) := \{ \M \in \upp{\pa} \mid \exists \pi \in \Prod({\mathcal{R}}) \text{~s.t.~} \pi \in \M \}$
\STATE $\mathit{Splitters} := \A(\mathcal{R}) \times \upp{\pa}(\mathcal{R})$
\WHILE{$\mathit{Splitters} \neq \emptyset$}
\STATE $(\rho, \M) \leftarrow \text{pop}(\mathit{Splitters})$
\FOR{$S \in \mathcal{S}$}
\STATE $S.\texttt{rr} \leftarrow 0$ 
\ENDFOR
\FOR{$\pi \in \M$}
\FOR{$(S + \rho \act{r} \pi) \in \pi.\texttt{inc}$}
\STATE $S.\texttt{rr} \leftarrow  S.\texttt{rr} + r$
\ENDFOR
\ENDFOR 
\STATE $\mathcal{H'} \leftarrow \text{split}(\mathcal{H}, \{ S.\texttt{rr} \mid S \in \mathcal{S}\})$\\
\IF{$\mathcal{H'} \neq \mathcal{H}$}
\STATE $\mathcal{H} \leftarrow \mathcal{H}'$
\STATE $\upp{\pa}(\mathcal{R}) := \{ \M \in \upp{\pa} \mid \exists \pi \in \Prod({\mathcal{R}}) \text{~s.t.~} \pi \in \M \}$
\STATE $\mathit{Splitters} := \A(\mathcal{R}) \times \upp{\pa}(\mathcal{R})$
\ENDIF
\ENDWHILE
\end{algorithmic}
\end{minipage}
\caption{Pre-processing (left); computation of the largest SE refining an initial partition (right).
}
\label{ls:preProceduresSMBANDproceduresSMB}
\end{figure}

\subsection{Existence and computation of the largest species equivalence}
\label{A.3}

%
We prove the existence of the largest \sbe\ by showing that the transitive closure of the union of two equivalence relations induced by an \sbe\ partition is still an \sbe.
To show this, we will use the notation $\eqpa{\pa}$ to denote the equivalence relation on $\Sp$ induced by the \sbe\ partition $\pa$. Formally,  given a  reaction network with species  $\Sp$ and reactions $\R$, a set of indices $I$, and an \se{} $\pa_i$  for all $i \in I$, we show that the transitive closure of their union $\eqpa{\pa} \ \equiv (\bigcup_{i \in I}\eqpa{\pa_i})^*$ induces 
an \se. We first note that $\eqpa{\pa}$ is an equivalence relation over $\Sp$ because it is the transitive closure of the union of equivalence relations over $\Sp$. For any $i \in I$, any block $\tilde\Y \in \pa_i$ is contained in a block $\Y \in \pa$, implying that any $\Y \in \pa$ is the union of blocks of $\pa_i$.
For any pair of equivalent species $(S_1,S_2) \in \ \eqpa{\pa}$, we have that $(S_1,S_2) \in (\bigcup_{i \in I}\eqpa{\pa_i})^n$, for some $n>0$,
where $(\bigcup_{i \in I}\eqpa{\pa_i})^n$ denotes the $n$-step transitive closure of the equivalence relations.

We now show that $\eqpa{\pa}$ is an \se{} by induction over $n$. Let   $\eqpa{{}}^n$ be $(\bigcup_{i \in I}\eqpa{\pa_i})^n$, and $\rho \in \ms(\mathcal{S})$. In the base case (i.e., $n=1$), we know that  $(S_1,S_2) \in \ \eqpa{{}}^1$ implies that $(S_1,S_2)\in \eqpa{\pa_i}$, for some $i\in I$. In order to prove that the condition required by \se{} holds, we use that for any $\Y\in \pa$ and any $i\in I$ we have that there exists some set of indices  $J^i$ such that $\Y=\bigcup_{j\in J^i}\tilde\Y_j$, with $\tilde\Y_j$ a block of $\pa_i$; hence, $\rr[S_1+\rho, \Y]=\sum_{j\in J^i}\rr[S_1+\rho, \tilde\Y_j]$. In the inductive step, we assume that the condition required by \se{} holds for $\eqpa{{}}^m$, $\forall m\!<\!n$. If $(S_1,S_2)\in\ \eqpa{{}}^n$, then there exists an $S_3\in\mathcal S$ such that
$(S_1,S_3)\in\ \eqpa{\pa_i}$ for some $i\!\in\!I$, and $(S_3, S_2)\in\ \eqpa{{}}^{n-1}$. Then, the claim follows from a similar argument as in the base case and the induction hypothesis.

\subsection{Computation of the largest species equivalence} 
\label{A.4}
Computing the largest \se\ can be encoded as a partition refinement problem~\cite{partitionref}, analogously to well-known algorithms for quantitative extensions of labeled transition systems in theoretical computer science~\cite{DBLP:journals/fuin/HuynhT92,DBLP:journals/jcss/BaierEM00,DBLP:conf/tacas/CardelliTTV16}. Hence, we only detail the conceptually novel parts. 

\paragraph*{Pre-processing} Throughout this section we assume that species that do not appear in any reaction are removed from the set of species $\mathcal{S}$; this can always be done in a pre-processing step.
Furthermore, we observe that the notion of reaction rate $\rr(\rho,\pi)$  is computed differently depending on whether $\rho\neq\pi$ or $\rho=\pi$. In particular,  the latter case is more complex to implement, as we have $\rr(\rho,\rho)=-\sum_{\pi'\neq\rho}\rr(\rho,\pi')$, requiring to consider all $\pi'$ different from $\rho$. Therefore, for a homogeneous and simpler treatment  we perform a preprocessing step that explicitly adds one \emph{self-loop} reaction $\rho\act{\rr(\rho,\rho)}\rho$ for each reagent $\rho\in\Reag(\R)$. This allows us to consider a simpler version of $\rr$ computed always according to the simpler case $\rho\neq\pi$. A pseudo-code for this is shown in Fig.~\ref{ls:preProceduresSMBANDproceduresSMB}(left), where we assume that each species $S\in\Sp$ is associated with a real-valued field, \texttt{rr}, used to compute reaction rates involving $S$ as reagent, and that each product $\pi\in\Prod(\mathcal{R})$ is provided with a  list, \texttt{inc}, which points to all the reactions that have $\pi$ as product.

Assuming that the list storing $\R$ is sorted according to a total lexicographical-like ordering on the reagents and products given by the ordering on species, then we have that
the preprocessing runs in $O(p\cdot r\cdot \log r)$ time. Indeed, we scan each reaction once. As discussed above, checking if the reagents of a reaction are the considered $\rho$ takes $O(p_r)$ time, 
as reagents and labels are stored as pairs (species,multiplicity) sorted according to a total ordering on species, and there are at most $p_r$ different species in the reagents and labels. If the reagents of the current reaction are equal to the currently considered reagents, we just add the reaction rate to \texttt{tot}, and the reaction to $\R'$, ignoring for the moment the sorting of $\R'$. Otherwise, we add to $\R'$ a new reaction with minus the computed cumulative reaction rate, an operation that takes constant time because reactions are stored as pointer data structures.
Once all reactions have been considered, we just have to sort the reactions in $\R'$, which takes $O(p\cdot r\cdot \log r)$ time.

\paragraph*{Algorithm} Our algorithm for computing the largest species equivalence of a reaction network is given in Fig.~\ref{ls:preProceduresSMBANDproceduresSMB}(right). 
At each iteration, every candidate partition $\pa$ (initialized with the input partition $\pa_0$) is associated with a set of splitters, consisting of pairs $(\rho,\M)$ 
where $\rho$ is any multiset that satisfies the condition in Eq.~1 in the main text, that is $\rho \in \A(\R) := \{\rho' \mid \exists S \in \mathcal{S} .\ \exists(S + \rho') \in \Reag(\R)\}$; $\M$ is any block of the multiset lifting $\upp{\pa}$ against which the condition Eq.~1 in the main text is checked, that is $\M \in \upp{\pa}(\mathcal{R}) := \{ \M \in \upp{\pa} \mid \exists \pi \in \Prod({\mathcal{R}}) \text{~such that~} \pi \in \M \}$.
The refinement of the partition occurs by computing the cumulative reaction rate $\sum_{\pi \in \mb} \rr(S + \rho, \pi)$, i.e., 
a side of Eq.~1 in the main text, for each splitter pair $(\rho,\M)$ and all species $S$. If all species in the same block have the same cumulative reaction rate for every splitter, then, by definition, the given candidate partition $\pa$ is indeed an \se\ and the algorithm terminates. Otherwise, if a splitter is such that a block of the candidate partition  contains species that have different cumulative reaction rates, then the block is refined into sub-blocks that have equal cumulative reaction rates. This leads to a further iteration of the algorithm that checks the new set of splitters arising from the as-refined candidate partition $\pa$.

We now analyze the time and space complexities of the algorithm.

\emph{Space complexity.}
 We assume that species and reactions are stored in data structures via pointers. The set of species $\cal S$ is stored as a list, while a block of species partition $\pa$ is a list of its species, each species in turn having a pointer to its block, requiring $O(s)$ space, where $s = | \mathcal{S} |$.
Also $\cal R$ is stored in a list of size $r = | \mathcal{R} |$. Each reaction consists of two lists in the form \texttt{(species, multiplicity)}, one for the reagents and one for the products, where the list for reagents is sorted according to a total ordering on species. 
Each list \texttt{inc} has size $O(r)$, 
while exactly $r$ entries appear in all \texttt{inc} lists.
Thus, storing $\cal R$ requires $O({p} \cdot {r})$ space, where $p$ is a bound on the maximum number of different species which have nonzero multiplicity across all reactions, that is, $p:=\max(p_r,p_p)$ where $p_r := \max \{ \sum_{S} \mathds{1}_{\{\rho(S) > 0\}} \mid \rho\in \Reag(\R)\}$ and $p_p := \max \{ \sum_S \mathds{1}_{\{\pi(S) > 0\}} \mid \pi\in\Prod(\mathcal{R})	\}$. We can bound $s$ by $O((p_r+p_p)\cdot r)=O(p\cdot r)$. This is because each reaction can have at most $p_r$ and $p_p$ different species as reagents and products, respectively. We observe that $p$ has upper bound equal to $s$; in practice, it ranged from $2$ to $4$ in all models herein considered.
Each element of $\A(\R)$ is not stored explicitly, but is represented implicitly by decreasing by 1 the multiplicity of a species in a multiset of reagents. Therefore, each element can be stored in constant space as a pair \texttt{(reagents,species)}, where \texttt{species} points to the species whose multiplicitly has to be assumed to be decreased by 1. For example, given the reagents $\rho=S_1+S_1+S_2$, we store $S_1 + S_2$ as $(\rho,S_1)$. 
An advantage of this representation is that we can compare two labels in $O(p_r)$ time even if different encodings are used for the same element. This is because reagents are lists of pairs \texttt{(species,multiplicity)} sorted with respect to the species.
Finally, $\A(\R)$ is stored in a sorted list too, requiring $O(l)$ space, where $l := |\A(\R)|$ which can be bound by
$O(p_r\cdot r)$, while insertions and searches cost $O(p_r\cdot \log l)$ time, because both operations require to compare $O(\log l)$ elements, each made by up to $p_r$ pairs.
In order to bound the size of the splitters to $O(\abs{\Prod(\mathcal{R})})=O(r)$  we do not explicitly store each pair $(\rho,\M)$. Instead we store only one, initialized with a reference to the first position of $\A(\R)$, and then update the pointer to the next position when necessary.
We store $\Prod(\mathcal{R})$ as a list,  requiring $O(r)$ space, while a partition of $\Prod(\mathcal{R})$ is encoded by representing a block with a list of pointers to its products, without worsening the space complexity. A partition of species is stored similarly.
In conclusion, the algorithm has an overall space complexity of $O(s + {p} \cdot {r})$, which can be bound by $O({p} \cdot {r})$.

\emph{Time complexity.}
Concerning time complexity, computing the partition $\upp{\pa}$ of $\Prod(\mathcal{R})$ according to the multi-set lifting of $\pa$ requires $O(s \cdot r \cdot (p_p+\log r))$ time, because it is
done by iteratively sorting the products in $\Prod(\mathcal{R})$ for each $\Y\in\pa$, according to the number of species in $\Y$ that they contain. There are at most $s$ blocks in $\pa$, and sorting the products costs $O(r\cdot p_p + r \cdot \log r)$ for each such block: for each element in $\Prod(\mathcal{R})$ it takes $O(p_p)$ time to count the number of species of $\Y$ in it, and sorting $\Prod(\mathcal{R})$ according to this value requires $O(r \cdot \log r)$ comparisons.
Then, a set \texttt{spls} of initial candidate splitters is generated for each $\rho \in \A(\R)$ and $\M \in \upp{\pa}$.

Computing the cumulative reaction rates is done by associating each  species $S$ with a real-valued field $S.$\texttt{rr} that is initialized to $0$, in $O(s)$ time.
Given a splitter $(\rho,\M)$, for each species $S$ the value $\rr[S+\rho,\M]$ is stored in $S.$\texttt{rr}  by iterating once the \texttt{inc} list of each $\pi \in \M$.
Checking for the presence of $\rho$ in the reagents of each reaction takes $O(p_r)$ time, since each multi-set is stored in a list sorted lexicographically.  Thus, the computation of the aggregate reaction rates has  $O(p_r\cdot r)$ time complexity, since each reaction appears in $\pi.$\texttt{inc} for one $\pi$ only.

Once the cumulative reaction rates are computed, the actual splitting is performed in the usual way, following, e.g.,~\cite{DBLP:journals/ipl/DerisaviHS03,DBLP:journals/jcss/BaierEM00}. It consists of the following three steps:
\begin{enumerate}
\item[(i)]
Each block is split using an associated balanced binary search tree (BST) in which each species $S$ of the block is inserted providing $\rr[S+\rho,\M]$ as key (stored in $S.$\texttt{rr}), and a new block is added to $\pa$ for each leaf of the BST; this requires $O(s\cdot \log s)$ time, as there are at most $s$ insertions in the BSTs, each having size at most $s$. BSTs do not worsen the space complexity, as only one for a block is built at a time.
%
\item[(ii)]
If at least one block has been split, all candidate splitters must be discarded; this takes $O(r)$ time, as \texttt{spls} contains at most an entry per product $\pi\in\Prod(\mathcal{R})$, since, for each block $\M$, only one entry is stored to represent all pairs $(\rho,\M)$); deletion from \texttt{spls} takes constant time assuming that it is implemented as a linked list.
\item[(iii)]
If at least a block has been split, all splitters have to be recomputed, which takes $O(s \cdot r \cdot (p_p+\log r))$ as previously discussed.
\end{enumerate}
In conclusion, overall the splitting procedure has time complexity
$O(s \cdot \log s + s\cdot r \cdot (p_p+\log r))$, which can be bound by $O(s \cdot r \cdot (p+\log r))$.
Indeed we have that:
\begin{align*}
O(s \cdot \log s + s\cdot r \cdot (p_p+\log r)) &= O(s \cdot \log s + s\cdot r \cdot p_p+ s\cdot r \cdot\log r)
\\
&= O(s (\log s + r \cdot p_p+ r \cdot\log r))
\\
&\leq O(s (r \cdot p+ r \cdot\log r))
\\
&= O(s\cdot r \cdot (p+ \log r))
\end{align*}
where the inequality follows from the fact that $O(\log s) < O(r\cdot p)$.

Finally, we observe that the splitting procedure  is invoked at most $O(l\cdot s\cdot r)$ times. This is because,
initially, $l\cdot r$ candidate splitters have to be considered. At every iteration where some blocks of $\pa$ are split (which happens at most $s$ times), all splitters are removed, and at most $l\cdot r$ new candidate ones are added to the set of splitters.
In conclusion, the overall computation of the largest \se\ takes $O(l\cdot s^2 \cdot r^2 \cdot (p +\log r))$ time and $O(p\cdot r)$ space.

\subsection{SIS model with heterogeneous rates} 
\label{A.5}
We consider a variant of the SIS model of the star network presented in \emph{Epidemic process in networks} and visualized in Fig. 4 in the main text. In this variant we assume node-dependent transmission and recovery rates that depend on whether the node is at the center or at the periphery of the star. More specifically, we consider the following mass-action \rn:
\begin{align*}
I_0 & \xrightarrow{\gamma_1} S_0  & I_1 & \xrightarrow{\gamma_2} S_1 &  I_2 & \xrightarrow{\gamma_2} S_2 &
I_3 & \xrightarrow{\gamma_2} S_3 & I_4 & \xrightarrow{\gamma_2} S_4 \notag \\
S_0 + I_1 & \xrightarrow{\beta_1} I_0 + I_1 & S_0 + I_2 & \xrightarrow{\beta_1} I_0 + I_2 &
S_0 + I_3 & \xrightarrow{\beta_1} I_0 + I_3 & S_0 + I_4 & \xrightarrow{\beta_1} I_0 + I_4 \label{eq:heterogeneous.sis} \\
S_1 + I_0 & \xrightarrow{\beta_2} I_1 + I_0 & S_2 + I_0 & \xrightarrow{\beta_2} I_2 + I_0 &
S_3 + I_0 & \xrightarrow{\beta_2} I_3 + I_0 & S_4 + I_0 & \xrightarrow{\beta_2} I_4 + I_0 \notag
\end{align*}
Using the parameter-independent network expansion method presented in the main text, it is possible to show that this model admits the same \se\ as in the main text, namely
$$\mathcal{H} = \big\{ \{S_0\}, \{I_0\}, \{S_1,S_2,S_3,S_4\}, \{I_1,I_2,I_3,I_4\} \big\},$$
for any $\gamma_1 \neq \gamma_2$ and $\beta_1 \neq \beta_2$.

\subsection{Comparison with syntactic Markovian bisimulation}
\label{A.6}

For completeness, we restate the definition of syntactic Markovian bisimulation (SMB) from~\cite{smbpaper} using the notation of this paper.
We start from the notion of reaction rate from~\cite{smbpaper}, which we call \emph{SMB-reaction rate} here in order to distinguish it from the main definition presented in this paper.

\begin{definition}[SMB-reaction rate (adapted from~\cite{smbpaper})]\label{def:SMBreactionRate}
Let $(\Sp,\R)$ be an \rn, and $\rho, \pi \in \MS(\Sp)$. The SMB-reaction rate from $\rho$ to $\pi$ is defined as
\[
\rr_{\mathit{SMB}}(\rho,\pi) =
\sum\limits_{(\rho \act{\rate} \pi) \in \R} \rate .
\]
For any $\M \subseteq \MS(\Sp)$, we define $\rr_{\mathit{SMB}}[\rho,\M] = \sum_{\pi\in \M} \rr_{\mathit{SMB}}(\rho,\pi)$.
\end{definition}

We now recall the notion of SMB.

\begin{definition}[SMB adapted from~\cite{smbpaper}]\label{def:smb}
Let $(\Sp,\R)$ be a reaction network, $\pa$ a partition of $\Sp$, and $\upp{\pa}$ its multiset lifting.
We say that $\pa$ is a syntactic Markovian bisimulation (SMB) for $(\Sp,\R)$ if and only if
\[\rr_{\mathit{SMB}}[S+\rho,\M] = \rr_{\mathit{SMB}}[S'+\rho,\M], \quad \text{for all~} \M,\tildeM \in \upp{\pa}, \text{all~} S,S'\in \tildeM, \text{and all~} \rho \in \MS(\Sp).\]
\end{definition}

We now provide a simple reaction network which shows that SMB is stricter than SE, the network consisting of the simple reaction  $S_1\act{1}S_2$.
We have that the partition consisting of only one block $\{S_1,S_2\}$ is an SE, but it is not an SMB.
Indeed, we have only one class of multi-set equivalent products, consisting of $\{S_1\}$ and $\{S_2\}$, with
\begin{align*}
\rr(S_1,\{\{S_1\},\{S_2\}\})&=0 \quad \text{ and }\quad \rr(S_2,\{\{S_1\},\{S_2\}\})=0
\\
\rr_{\mathit{SMB}}(S_1,\{\{S_1\},\{S_2\}\})&=1 \quad \text{ and }\quad \rr_{\mathit{SMB}}(S_2,\{\{S_1\},\{S_2\}\})=0
\end{align*}

\begin{table}[t]
\centering
\caption{Comparison of species equivalence (SE) with syntactic Markovian bisimulation (SMB).}
\label{smb_table1}
\scalebox{1.0}
{
\begin{tabular}{ccc rrHr rrHr}
\toprule
&&&\multicolumn{4}{c}{\emph{Number of species}} & \multicolumn{4}{c}{\emph{Number of reactions}} \\
\cmidrule(lr){4-7} \cmidrule(lr){8-11}
 \textit{Id} & \textit{Model} & \emph{Ref.} & \multicolumn{1}{c}{\textit{Orig.}} & \multicolumn{1}{c}{\textit{SE}} & \multicolumn{1}{H}{\textit{SE-cur}} & \multicolumn{1}{c}{\textit{SMB}} & \multicolumn{1}{c}{\textit{Orig.}} & \multicolumn{1}{c}{\textit{SE}} & \multicolumn{1}{H}{\textit{SE-cur}} & \multicolumn{1}{c}{\textit{SMB}}  \\ \midrule
 1  & fceri\_fyn\_trimer & \cite{Faeder01042003,citeulike:8493139} & \np{20881} & 834   & {834}   &  \np{20881}  & \np{407308} &  \np{7620}  &  {\np{7620}}  &  \np{407308} \\
 2  & fceri\_gamma2\_asym & \cite{Faeder01042003,citeulike:8493139} & \np{10734} & 351   & {351}   &  \np{3744}  & \np{187468} &  \np{2532}  &  {\np{2532}}  &  \np{50368} \\  
 3  & fceri\_fyn & \cite{Faeder01042003,citeulike:8493139} &  \np{1281}  & 154   & {154}   & { \np{1281}} &  \np{15256} & 900   & 900   & \np{15256} \\
 4  & Nag2009 & \cite{nag2009aggregation} & 920   & 364   & {364}   & {920} &  \np{12740} &  \np{3420}  &  {\np{3420}}  &  \np{12740}  \\ 
 5  & fceri\_lyn\_745 & \cite{Faeder01042003,citeulike:8493139} & 745   & 105   & {105}   & {745} & \np{ 8620}  & 576   & {576}   & \np{8620} \\
 6  & fceri\_ji & \cite{Faeder01042003,citeulike:8493139} & 354   & 105   & {105}   & {354} &  \np{3680}  & 576   & {576}   & { \np{3680}} \\
 7  & LipidRafts  
& \cite{barua2012mechanistic} & 348   & 215   & {215}   & {348} &  \np{3447}  &  \np{1782}  &  {\np{1782}}  & { \np{3447}} \\ 
 8 & NIHMS80246-S4 & \cite{borisov2008domain} & 213   & 66 & {66} & {213} &  \np{2230}  & 432   & {432}   & { \np{2230}} \\ 
 9 & NIHMS80246-S6 & \cite{borisov2008domain} & 24 & 3  & 3 & {24} & 88 & 2  & {2}  & {88} \\ 
\bottomrule
\end{tabular}
}
\end{table}

Table~\ref{smb_table1} shows  a number of models from the literature where, in practice, \se\ can aggregate more than syntactic Markovian bisimulation.

\subsection{Speeding up stochastic simulations with species equivalence}
\label{A.7}
We discuss how \se\ can reduce the runtimes of stochastic simulation algorithms.
In order to perform these tests on state-of-art stochastic simulation algorithms,  we used the StochKit simulation framework~\cite{DBLP:journals/bioinformatics/SanftWRFLP11}, performing experiments with the implementations of the original direct method (SSA) by Gillespie~\cite{Gillespie77},
the next-reaction method (NRM) by Gibson and Bruck~\cite{doi:10.1021/jp993732q},
 as well as the more recent Logarithmic Direct Method (LDM)~\cite{Li06logarithmicdirect} and Composition and Rejection (CR)~\cite{doi:10.1063/1.2919546}. The performances of such algorithms have been already compared, e.g.~\cite{doi:10.1063/1.2919546,Li06logarithmicdirect}.

For each simulation algorithm, the speed-up was measured as the ratio between the runtimes of 5 independent simulation of the original and the reduced networks, using the same time horizons and initial conditions provided in the original articles from which the models have been taken. This speed-up metric does not include the time to compute the reduced network by \se, which however turned out to be negligible because it took at most one twentieth of the analysis time of the reduced model.

\begin{table}[t]
	\caption{Analysis speed-ups for the models in Table~\ref{smb_table1} with different simulation algorithms.}
	\label{speedups_table}
\centering
	\scalebox{1.0}
	{
		\begin{tabular}{cHr r rrrr}
			\toprule
			\multicolumn{4}{c}{\emph{Model}} &\multicolumn{4}{c}{\emph{Speedup ratios original/reduced}}  \\
			\cmidrule(lr){1-4}
			\cmidrule(lr){5-8}
			\textit{Id} &\textit{Model} & \multicolumn{1}{c}{\emph{Horizon}} &
			\multicolumn{1}{c}{\emph{SE (s)}} &
			\multicolumn{1}{c}{\emph{SSA}} & \multicolumn{1}{c}{\emph{NRM}} & \multicolumn{1}{c}{\emph{LDM}} & \multicolumn{1}{c}{\emph{CR}}
			\\ \midrule
			1 & fceri\_fyn\_trimer & \np{3840}
			& 1.51E+1
			& \multicolumn{1}{c}{$\geq$\textit{654.2}} & \multicolumn{1}{c}{$\geq$\textit{971.1}} & \multicolumn{1}{c}{$\geq$\textit{187.4}} & \multicolumn{1}{c}{$\geq$\textit{\np{1310}.5}}
			\\
			2 & fceri\_gamma2\_asym & \np{3840}
			& 2.26E+0
			& \np{6730}.2 & 144.0 & \np{1070}.5 & 645.1
			\\
			3 & fceri\_fyn & \np{3840}
			& 2.69E--1
			& 33.9 & 6.2 & 8.2 & 4.7
			\\
			4 & Nag2009 & \np{200}
			& 3.26E--1
			& 3.4 & 2.6 & 3.1 & 2.3
			\\
			5 & fceri\_lyn\_745 & \np{3840}
			& 6.00E--2
			& 53.2 & 6.3 & 9.7 & 5.1
			\\
			6 & fceri\_ji & \np{3840}
			& 2.30E--2
			& 20.8 & 5.0 & 6.3 & 4.3
			\\
			7 & LipidRafts 
			& \np{3600}
			& 7.00E--2
			& 1.3 & 2.0 & 2.9 & 1.6
			\\
			8 & NIHMS80246-S4 & \np{40}
			& 1.50E--2
			& 4.3 & 3.9 & 3.8 & 3.1
			\\
			9 & NIHMS80246-S6 & \np{40}
			& 1.00E--3
			& 4.0 & 6.7 & 4.3 & 5.2
			\\
			\bottomrule
		\end{tabular}
	}
\end{table}

Table~\ref{speedups_table} shows the speed-up results by  SE on the models from Table~\ref{smb_table1}. As an indicator of the cost of the reduction, the third column (\emph{SE}) shows the execution times of SE as measured on commodity hardware (a laptop with 8\,GB RAM and a 3,1 GHz Dual-Core Intel Core i5). For the first model we report a lower bound on the speed-up because a single simulation of the original model did not terminate before \np{20000}\,s. An inspection of the cause of such a large execution time revealed that the simulation engine allocated  more memory than available, leading to frequent memory swaps that significantly degraded performance. Instead, upon reduction the simulation of the the same model took a few minutes on our machine.

The results also indicate that larger speed-ups can be achieved with the larger models of our dataset, which are also the more computationally demanding for stochastic simulations. As expected, the larger speed-ups are obtained when using the direct SSA method, however significant improvements of two-three orders of magnitude can be reported for model efficient methods such as NRL and LDM, which are designed to provide logarithmic time dependency on the number of reactions, or CR, which can offer a constant time dependency under certain conditions~\cite{doi:10.1063/1.2919546}.

\end{document}